\documentclass[a4paper,UKenglish]{lipics-v2019}
\pdfoutput=1

\usepackage[dvipsnames]{xcolor}

\usepackage{amsmath,amssymb,amsthm,graphicx,fixmath,tikz,latexsym,xspace}
\usetikzlibrary{plotmarks}
\usetikzlibrary{shapes}

\newcommand{\A}{\mathrm{A}}

\title{Hardness and Approximation of Minimum Convex Partition} 

\titlerunning{Hardness and Approximation of Minimum Convex Partition} 

\author{Nicolas Grelier}{Department of Computer Science, ETH Z\"urich, Switzerland}{nicolas.grelier@inf.ethz.ch}{}{}

\authorrunning{N. Grelier} 

\Copyright{Nicolas Grelier} 



\category{} 

\relatedversion{} 

\supplement{}

\funding{Research supported by the Swiss National Science Foundation within the collaborative DACH project \emph{Arrangements and Drawings} as SNSF Project 200021E-171681.}

\acknowledgements{The author thanks Michael Hoffmann for his helpful advice.}

\nolinenumbers 

\hideLIPIcs  

\EventEditors{John Q. Open and Joan R. Access}
\EventNoEds{2}
\EventLongTitle{42nd Conference on Very Important Topics (CVIT 2016)}
\EventShortTitle{CVIT 2016}
\EventAcronym{CVIT}
\EventYear{2016}
\EventDate{December 24--27, 2016}
\EventLocation{Little Whinging, United Kingdom}
\EventLogo{}
\SeriesVolume{42}
\ArticleNo{23}

\begin{document}

\maketitle

\begin{abstract}
We consider the Minimum Convex Partition problem: Given a set $P$ of $n$ points in the plane, draw a plane graph $G$ on $P$, with positive minimum degree, such that $G$ partitions the convex hull of $P$ into a minimum number of convex faces. We show that Minimum Convex Partition is NP-hard, and we give several approximation algorithms, from an $\mathcal{O}(\log \textit{OPT})$-approximation running in $\mathcal{O}(n^8)$-time, where $\textit{OPT}$ denotes the minimum number of convex faces needed, to an $\mathcal{O}(\sqrt{n}\log n)$-approximation algorithm running in $\mathcal{O}(n^2)$-time. We say that a point set is $k$-directed if the (straight) lines containing at least three points have up to $k$ directions. We present an $\mathcal{O}(k)$-approximation algorithm running in $n^{\mathcal{O}(k)}$-time. Those hardness and approximation results also holds for the Minimum Convex Tiling problem, defined similarly but allowing the use of Steiner points. The approximation results are obtained by relating the problem to the Covering Points with Non-Crossing Segments problem. We show that this problem is NP-hard, and present an FPT algorithm. This allows us to obtain a constant-approximation FPT algorithm for the Minimum Convex Partition Problem where the parameter is the number of faces.
\end{abstract}


\newpage

\section{Introduction}

The CG Challenge 2020 organised by Demaine, Fekete, Keldenich, Krupke and Mitchell~\cite{CGshop}, was about solving instances of \emph{Minimum Convex Partition} (MCP). 

\begin{definition}[Demaine {\em et al.}~\cite{CGshop}: Minimum Convex Partition problem]
\normalfont Given a set $P$ of $n$ points in the plane. The objective is to compute a plane graph with vertex set $P$ (with each point in $P$ having positive degree) that partitions the convex hull of $P$ into the smallest possible number of convex faces. Note that collinear points are allowed on face boundaries, so all internal angles of a face are at most $\pi$.
\end{definition}

As explained by Bose {\em et al.}, this problem has applications in routing~\cite{bose2002online}. They showed that a routing algorithm named \emph{Random-Compass} that works for triangulations can be extended to convex partitions. Having a convex partition with few faces reduces the amount of data to store. From now on, we denote by $P$ a set of $n$ points in the plane.

In this paper, we present several approximation algorithms for MCP. We obtain those approximation algorithms by relating the MCP problem to the \emph{Covering Points with Non-Crossing Segments} (CPNCS) problem. First, we define what \emph{non-crossing segments} are.
 
 \begin{definition}[Non-Crossing Segments]
 \normalfont We call a part of a (straight) line bounded by two points a \emph{segment}. The two points are referred to as \emph{endpoints} of the segment. Note that we do not force the endpoints to be distinct, therefore we consider a point $p$ as being a segment. The endpoint of $p$ is $p$ itself. Two segments are \emph{non-crossing} if the intersection of their relative interior is empty.
 \end{definition}
 
  \begin{definition}[Covering Points with Non-Crossing Segments]
 \normalfont Given a set $P$ of $n$ points, find a minimum number of non-crossing segments whose endpoints are in $P$ such that each point of $P$ is contained in at least one segment.
 \end{definition}
 
The condition that the endpoints of the segments must be in $P$ has no effect on the number of segments required. We add it as it simplifies some arguments. Note that CPNCS is not a so-called \emph{set cover problem} nor an \emph{exact cover problem}. We believe that CPNCS is interesting in itself. Even though it is a very natural problem, to the best of our knowledge it had not been introduced before.

 \subsection{NP-hardness results}
 
 Fevens, Meijer and Rappaport first considered the MCP problem in 2001~\cite{fevens2001minimum}, and its complexity was explicitly asked about by Knauer and Spillner in 2006~\cite{knauer2006approximation}. It has remained open since then~\cite{barboza2019minimum, CGshop}. We show in Section~\ref{sec:NPMCP} that MCP is NP-hard. To do this, we use the decision version of the problem, as stated below:
 
 \begin{definition}[MCP - decision version]
\label{def:MCP}
\normalfont Given a set $P$ of points in the plane and a natural number $k$, is it possible to find at most $k$ closed convex polygons whose vertices are points of $P$, with the following properties:
\begin{itemize}
    \item The union of the polygons is the convex hull of $P$,
    \item The interiors of the polygons are pairwise disjoint,
    \item No polygon contains a point of $P$ in its interior.
\end{itemize}
\end{definition}
 
 We also show NP-hardness of a similar problem, which we call \emph{Minimum Convex Tiling} problem (MCT). The problem is exactly as in Definition~\ref{def:MCP}, but the constraint about the vertices of the polygons is removed (i.e. they need not be points of $P$). This can make a difference as shown in Figure~\ref{fig:MCPneqMCT}. Equivalently, the MCT problem corresponds to the MCP problem when Steiner points are allowed. A \emph{Steiner point} is a point that does not belong to the point set given as input, and which can be used as a vertex of some polygons. The MCT problem has been studied in 2012 by Dumitrescu, Har-Peled and T{\'o}th, who asked about the complexity of the problem~\cite{dumitrescu2012minimum}. We answer their question, and our proofs are very similar for MCP and MCT.

\begin{figure}[h]
    \centering
    \begin{tikzpicture}[scale=1]

\draw[fill, orange, opacity=0.3] (-2,0)--(0,2)--(2,3)--(2,0)--cycle;
\draw[fill, red, opacity=0.3] (-2,0)--(0,2)--(-2,3)--cycle;
\draw[fill, blue, opacity=0.3] (2,3)--(0,2)--(-2,3)--cycle;

\node at (0,2) {\textbullet};
\node at (-2,0) {\textbullet};
\node at (-2,3) {\textbullet};
\node at (2,0) {\textbullet};
\node at (2,3) {\textbullet};

\draw[fill, red, opacity=0.3] (5,0)--(7,0)--(7,3)--(5,3)--cycle;
\draw[fill, blue, opacity=0.3] (9,0)--(7,0)--(7,3)--(9,3)--cycle;

\node at (7,2) {\textbullet};
\node at (5,0) {\textbullet};
\node at (5,3) {\textbullet};
\node at (9,0) {\textbullet};
\node at (9,3) {\textbullet};

    \end{tikzpicture}
    \caption{A minimum partition with three convex polygons and a tiling with two.}
    \label{fig:MCPneqMCT}
\end{figure} 

We show in Section~\ref{sec:NPhardCPNCS} that CPNCS is NP-hard, even for some constrained point sets, using a reduction from \emph{Maximum Independent Set in Intersection Graphs of Segments}.

\subsection{Approximation algorithms}
For the related problem \emph{Minimum Convex Partition of Polygons with Holes}, Bandyapadhyay, Bhowmick and Varadarajan showed the existence of a $(1+\varepsilon)$-approximation algorithm running in time $n^{\mathcal{O}((\log n / \varepsilon)^4)}$~\cite{bandyapadhyay2014approximation}. Although they only consider holes with non empty interior, one can observe that their proof extends to the case of point holes. This is an even more general setting than MCP for point sets, so their algorithm also applies in our setting. This implies that MCP is not APX-hard unless $\textit{NP}\subseteq \textit{DTIME}(2^{\mathrm{polylog}~n})$.

Under the assumption that no three points are collinear, Knauer and Spillner have shown a $\frac{30}{11}$-approximation algorithm~\cite{knauer2006approximation} for MCP in 2006. As a lower bound on the number of convex faces for one particular point set, they rely on the observation that each inner point has degree at least $3$. This gives a lower bound on the number of edges, and therefore on the number of faces, by Euler's formula. Note that the restriction that no three points are on a line is necessary, as shown in Figure~\ref{fig:noLowerBound}. There are only two faces in a minimum convex partition of this point set, and all the inner points have degree $2$.
  
\begin{figure}[h]
    \centering
    \begin{tikzpicture}[scale=0.7]

\draw[fill, red, opacity=0.3] (2.32,0) -- (1.3,1.11)-- (0.16,1.73)-- (-2,1.27) --(-2.78,-1.19)  --cycle;
\draw[fill, blue, opacity=0.3] (2.32,0) --  (2.38,-2.25) --(-0.44,-2.69)--(-2.78,-1.19)  --cycle;

\node at  (1.3,1.11) {\textbullet};
\node at  (0.16,1.73) {\textbullet};
\node at  (-2,1.27) {\textbullet};
\node at  (-2.78,-1.19) {\textbullet};
\node at  (-0.44,-2.69) {\textbullet};
\node at  (2.38,-2.25) {\textbullet};
\node at(2.32,0) {\textbullet};
\node at(-1.9982297154899895,-1.0075869336143308) {\textbullet};
\node at(-1.1974920969441518,-0.8207481559536354) {\textbullet};
\node at(-0.6398524762908324,-0.6906322444678609) {\textbullet};
\node at(0.08849315068493152,-0.5206849315068494) {\textbullet};
\node at(0.995131717597471,-0.3091359325605901) {\textbullet};

    \end{tikzpicture}
    \caption{The number of inner points can be arbitrarily much larger than the number of convex faces required.}
    \label{fig:noLowerBound}
\end{figure}
  
Additionally, Knauer and Spillner showed how to adapt any constructive upper bound on the number of faces into an approximation algorithm. More explicitly, they showed that if one can compute in polynomial time a convex partition with at most $\lambda n$ convex faces, then there exists a $2\lambda$-approximation algorithm running in polynomial time. The best result to date is a proof by Sakai and Urrutia that one can partition a point set in quadratic time using at most $\frac{4}{3}n$ convex faces (the result was presented at the 7th JCCGG in 2009, the paper appeared on arXiv in 2019)~\cite{sakai2019convex}. Although they do not mention it, combining this result with the one by Knauer and Spillner gives a quadratic time $\frac{8}{3}$-approximation algorithm.
  
Concerning previous upper bounds, Neumann-Lara, Rivero-Campo and Urrutia first showed in 2004 how to construct in quadratic time a partition of any point set with at most $\frac{10}{7}n$ convex faces~\cite{neumann2004note}. In 2006, Knauer and Spillner improved this to $\frac{15}{11}n$ convex faces~\cite{knauer2006approximation}. As said above, the best known upper bound is $\frac{4}{3}n$, as proven by Sakai and Urrutia in 2009.
  
Relatedly for lower bounds, Garc{\'\i}a-Lopez and Nicol{\'a}s have given in 2013 a construction of point sets for which any convex partition has at least $\frac{35}{32}n-\frac{3}{2}$ faces~\cite{garcia2013planar}. 
 
 All these results concerning upper bounds hold for all point sets, even where many points are on a line. Indeed, slightly shifting the points so that no three points are on a line can only increase the number of convex faces needed. So an upper bound for point sets where no three points are on a line also holds for all point sets. However, as mentioned above, the lower bound used by Knauer and Spillner does not extend to our setting, where we consider all point sets. They say that a constant-approximation algorithm would be desirable for unrestricted point sets, but so far not even an $\mathcal{O}(n^{1-\varepsilon})$-approximation is known. For the MCT problem, Dumitrescu, Har-Peled and T{\'o}th showed the existence of a $3$-approximation algorithm for point sets with no three collinear points~\cite{dumitrescu2012minimum}. They also ask whether a constant-approximation algorithm exists when this constraint is removed. However, so far no $\mathcal{O}(n^{1-\varepsilon})$-approximation algorithm is known. In Section~\ref{sec:approxCPNCS}, we prove the following:
 
 \begin{theorem}\label{thm:approxAlgo}
 There exist $\mathcal{O}(\log \textit{OPT})$-approximation algorithms for MCP, MCT and CPNCS running in $\mathcal{O}(n^8)$-time.
 \end{theorem}
 
 Allowing several points to be on a line does not simply create tedious technicalities to deal with. The crux of the matter is to find, for a fixed point set, an exploitable lower bound on the number of faces in a minimum convex partition. When no three points are on a line, the number of inner points in $P$ gives a linear lower bound on the number of faces in a convex partition~\cite{knauer2006approximation}, and in a convex tiling~\cite{dumitrescu2012minimum}. In this paper, we consider point sets with no restriction. We introduce the CPNCS problem as it pinpoints where the difficulty of finding a constant-approximation algorithm for MCP is and makes the problem easier to study. The \emph{inner points} of $P$ are the points not on the boundary of the convex hull. We show in Section~\ref{sec:MCPCPNCS} the following:
 
 \begin{theorem}\label{thm:MCPEquivalentCPNCS}
 Let $P$ be a set of $n$ points with at least one inner point, and let $\lambda\geq1$ be a real number. Let $f_m$ denote the minimum number of faces in a convex partition of $P$. Let $s_m$ denote the minimum number of non-crossing segments in a covering of the inner points of $P$, denoted by $P_i$. 
 \begin{enumerate}
     \item It holds that $\frac{s_m}{6}\leq f_m \leq 8 s_m$.
     \item Given a covering of $P_i$ with $s\leq \lambda s_m$ non-crossing segments, it is possible to compute in $\mathcal{O}(n^2)$-time a convex partition of $P$ with at most $24\lambda f_m$ convex faces.
     \item Given a convex partition of $P$ with $f\leq \lambda f_m$ convex faces, it is possible to compute in $\mathcal{O}(n)$-time a covering of $P_i$ with at most $44\lambda s_m$ non-crossing segments.
 \end{enumerate}
The theorem also holds when considering convex tilings instead of convex partitions.
 \end{theorem}
 
 \begin{figure}[ht]
    \centering
    \begin{tikzpicture}[scale=1]
    
  \draw[red,thick] (2,1)--(2,6);    
    \draw[red,thick] (-0.3,1)--(0.7,2); 
      \draw[red,thick] (-1.5,3)--(1.5,4); 
        \draw[red,thick] (-0.6,5)--(0.3,6); 
        
\draw[red,thick] (6.7,1)--(9,1);
\draw[red,thick] (7.7,2)--(9,2);
\draw[red,thick] (5.5,3)--(9,3);
\draw[red,thick] (8.5,4)--(9,4);
\draw[red,thick] (6.4,5)--(9,5);
\draw[red,thick] (7.3,6)--(9,6);
    
\draw (-2,2)--(-1.5,3);
\draw (-0.3,1)--(2,0);
\draw (-1.5,3)--(-0.6,5)--(-2,7);
\draw (-2,7)--(0.3,6)--(2,7);
\draw (1.5,4)--(2,5);
\draw (1.5,4)--(0.7,2)--(2,2);
\draw (-2,1)--(-0.3,1);
\draw (2,0)--(3,0)--(3,7)--(2,7);
\draw (2,0)--(-2,0)--(-2,7)--(2,7);
\draw (2,6)--(2,7);
\draw (2,0)--(2,1);
    
\foreach \y in {0,...,7}{
\node at (-2,\y) {\textbullet};
\node at (2,\y) {\textbullet};
\node at (3,\y) {\textbullet};
\node at (5,\y) {\textbullet};
\node at (9,\y) {\textbullet};
\node at (10,\y) {\textbullet};
\draw (9,\y)--(10,\y);}   

\node at (-0.3,1) {\textbullet};
\node at (0.7,2) {\textbullet};
\node at (-1.5,3) {\textbullet};
\node at (1.5,4) {\textbullet};
\node at (-0.6,5) {\textbullet};
\node at (0.3,6) {\textbullet};

\node at (6.7,1) {\textbullet};
\node at (7.7,2) {\textbullet};
\node at (5.5,3) {\textbullet};
\node at (8.5,4) {\textbullet};
\node at (6.4,5) {\textbullet};
\node at (7.3,6) {\textbullet};

\draw (9,7)--(5,7)--(5,0)--(9,0);
\draw (10,7)--(10,0);
\draw (5,1)--(6.7,1);
\draw (5,2)--(7.7,2);
\draw (5,3)--(5.5,3);
\draw (5,4)--(8.5,4);
\draw (5,5)--(6.4,5);
\draw (5,6)--(7.3,6);

    \end{tikzpicture}
    \caption{On the left side, a minimum covering of the inner points of $P$ with $4$ segments. A convex partition which contains those segments has at least $9$ convex faces. On the right side, a covering of the inner points of $P$ with $6$ segments. There exists a minimum convex partition of $P$ with $7$ faces, which contains those segments.}
    \label{fig:MCPneqCPNCS}
\end{figure}
 
 \begin{remark}
The idea behind the similarity of MCP, MCT and CPNCS is that they are all about maximizing the number of vertices of degree $2$ with incident edges being aligned in a plane straight-line drawing of a graph on a point set. We illustrate in Figure~\ref{fig:MCPneqCPNCS} that MCP and CPNCS are not strictly equivalent: We give a point set $P$ in which the unique minimum convex partition of $P$ does not contain the segments of any minimum covering of the inner points of $P$ with non-crossing segments. Let us denote by $P_s$ the set of endpoints of the three non-vertical segments of the covering. In any convex partition of $P$ that contains those four segments of the covering, the points in $P_s$ need to have degree at least $3$. Moreover, each of them has to be connected to at least one point not in $P_s$. Finally, the topmost and bottommost of those six endpoints must be connected to at least two points not in $P_s$. Therefore, the convex partition drawn on the left is one that minimises the number of edges (and thus of faces) among the convex partitions that contain this minimum covering. This implies that finding a minimum covering of the inner points of some point set $P$ with non-crossing segments does not necessarily help in finding a minimum convex partition of $P$. Nonetheless, Theorem~\ref{thm:MCPEquivalentCPNCS} states that such a covering leads to an approximation for the MCP problem.
\end{remark}

 We call the algorithm for CPNCS the algorithm that iteratively picks a new segment among the valid ones that cover as many points not yet covered as possible, until all points in $P$ are covered, the \emph{greedy} algorithm. As we consider points to be potential segments, the algorithm terminates. We prove the following in Section~\ref{sec:badPointSet}.
 
 \begin{theorem}\label{thm:lowerBoundGreedy}
 There exist point sets for which the greedy algorithm for solving CPNCS realises an $\Omega(\sqrt{n})$-approximation.
 \end{theorem}

 



The CPNCS problem bears a resemblance with \emph{Covering Points with Lines} (CPL), defined below, that we use in one of our approximation algorithms. 
 
\begin{definition}[Covering Points with Lines]
\normalfont Given a set $P$ of $n$ points, find a minimum number of lines such that each point of $P$ is contained in at least one line.
\end{definition}

Before going into the proofs, we want to make a remark that we deem interesting. In~\cite{kumar2000hardness}, Anil Kumar, Arya and Ramesh mention that the CPL problem was motivated by the problem of \emph{Covering a Rectilinear Polygon with Holes using Rectangles}. They say that getting a $o(\log n)$-approximation for this problem seems to require a better understanding of CPL. However, they are ``not sure of the exact nature of this relationship''. In this paper, we show the hardness of MCP by using tools developed by Lingas to show NP-hardness of \emph{Minimum Rectangular Partition for Rectilinear Polygons with Holes}~\cite{lingas1982power}. The difference between a covering and a partition is that, in the latter, objects are interior-disjoint. Moreover, we prove that obtaining a constant-approximation algorithm for MCP is equivalent to finding one for CPNCS. Again, CPNCS is the non-crossing version of CPL. We hope this paper helps to better understand the relationship between these problems.

 \subsection{Exact algorithms, FPT algorithms}
 
Under the assumptions that the points lie on the boundaries of a fixed number $h$ of nested convex hulls, and that no three points lie on a line, Fevens, Meijer and Rappaport gave an algorithm for solving MCP in time $\mathcal{O}(n^{3h+3})$~\cite{fevens2001minimum}. Observe that this is not an FPT algorithm. Some integer linear programming formulations of the problem have been recently introduced~\cite{barboza2019minimum,sapucaia2021solving,cambazard2021integer}.
  
A first FPT algorithm with respect to the number $k$ of inner points was introduced by Grantson and Levcopoulos, with running time $\mathcal{O}(2^{16k}k^{6k-5}n)$~\cite{grantson2004fixed}. The idea of the algorithm is to enumerate all plane graphs on the inner points, and then for each to them to guess how to connect the inner points to the points on the boundary of the convex hull. Another FPT algorithm with respect to the number of inner points was later found by Spillner, with running time $\mathcal{O}(2^kk^4n^3+n \log n)$~\cite{spillner2008fixed}.
 
We show in Section~\ref{sec:FPT} the existence of an FPT algorithm that checks whether there is a solution for CPNCS with at most $k$ non-crossing segments, running in time $\mathcal{O}(2^{k^2}k^{7k}+n^4 \log n)$. By Theorem~\ref{thm:MCPEquivalentCPNCS}, this gives us a constant-approximation FPT algorithm for MCP and MCT, where the parameter is the number of convex faces needed. Under the assumption that no three points are on a line, the number of faces in a minimum convex partition or in a minimum convex tiling is the same as the number of inner points, up to a constant multiplicative factor~\cite{knauer2006approximation, dumitrescu2012minimum}. However, when removing this assumption, the number of inner points can be arbitrarily much larger than the minimum number of convex faces, as shown in Figure~\ref{fig:noLowerBound}. Our algorithm runs in time $\mathcal{O}(2^{36f^2}f^{42f+1}+n^4 \log n)$, where $f$ denotes the minimum number of convex faces needed in a convex partition or in a convex tiling.

\section{The relation between MCP, MCT and CPNCS}
\label{sec:MCPCPNCS}

Throughout this section, we denote by $P$ a point set in the plane. We denote by $P_i$ the set of inner points of $P$. Let $p$ be in $P$. If $P$ and $P\setminus\{p\}$ do not have the same convex hull, we say that $p$ is an \emph{extreme point}. We denote by $P'$ the extreme points in $P_i$. Note that a point might lie on the boundary of the convex hull of a point set without being an extreme point. We say that $P$ is \emph{special} if $|P'|\leq 2$. Recall that for a given covering of a point set $Q$ with non-crossing segments, we always assume that the endpoints of the segments are in $Q$.

\begin{lemma}\label{lemma:segmentToPartition}
Let $P$ be a set of $n$ points that is not special. Given a covering $K$ of $P_i$ with $s$ non-crossing segments, one can compute in $\mathcal{O}(n^2)$-time a convex partition $\Sigma$ of $P$ with at most $4s+2|P'|$ faces. Moreover a segment in $K$ is the union of some edges in $\Sigma$.
\end{lemma}

\begin{proof}
Let $Q\subseteq P_i$ be the set of the endpoints of segments in the covering. Note that $|Q|$ is at most $2s$. As $P$ is not special, there exist triangulations of $Q$. We compute a constrained triangulation (for example Delaunay) of $Q$ with respect to the segments of the covering. This can be done in $\mathcal{O}(n \log n)$-time~\cite{chew1989constrained}, and there are at most $2|Q|$ faces. We observe that the triangulation of $Q$ gives a convex partition of $P_i$. We add all segments between consecutive points on the boundary of the convex hull of $P$. Now, it remains to deal with the surface that is within the convex hull of $P$, but not within the convex hull of $P_i$. To do that, we add for each point in $P'$ at most two edges to points on the boundary of the convex hull of $P$. We do it such that the angle between any consecutive edges around a point in $P'$ is at most $\pi$. This takes $\mathcal{O}(n^2)$ time~\cite{knauer2006approximation}. We have now obtained a convex partition of~$P$.
\end{proof}

If one is interested in a convex tiling instead of a convex partition in Lemma~\ref{lemma:segmentToPartition}, note that it is possible to add only one edge for each point in $P'$, resulting in a convex tiling with at most $4s+|P'|$ faces.

\begin{lemma}\label{lemma:partitionToSegment}
Let $P$ be a set of $n$ points. Given a convex tiling $\Sigma$ of $P$ with $f$ faces, one can compute in $\mathcal{O}(n)$-time a covering $K$ of $P_i$ with at most $6f-2|P'|$ non-crossing segments. Moreover a segment in $K$ is the union of some edges in $\Sigma$.
\end{lemma}

\begin{figure}
    \centering
    \begin{tikzpicture}[scale=0.5]
\draw[blue, thick]  (-0.78,5.26)--(1.59,-3.54); 
\draw  (-4.19,2.72)-- (1.63,7.06);
\draw  (1.63,7.06)-- (10.39,2.38);
\draw  (10.39,2.38)-- (10.35,-2.76);
\draw  (10.35,-2.76)-- (3.23,-6.68);
\draw  (3.23,-6.68)-- (-6.25,-3.6);
\draw  (-6.25,-3.6)-- (-4.19,2.72);
\draw[OliveGreen,dashed,thick]  (-2.330311907110875,4.106777718752372)-- (0.40888402202969454,-5.763434893233276);
\draw[blue,thick] (-4.19,2.72)-- (-3.33,-2.08);
\draw[red,thick]  (-3.33,-2.08)-- (-6.25,-3.6);
\draw[blue,thick]  (-3.33,-2.08)-- (0.40888402202969454,-5.763434893233276);
\draw[red,thick]  (5.37,0.2)-- (10.39,2.38);
\draw  (5.37,0.2)-- (2.23,3.12);
\draw  (2.23,3.12)-- (1.59,-3.54);
\draw[red,thick]  (2.23,3.12)-- (1.63,7.06);
\draw[red,thick]  (1.59,-3.54)-- (0.40888402202969454,-5.763434893233276);
\draw  (5.37,0.2)-- (6.31,-3.96);
\draw[blue,thick]  (6.31,-3.96)-- (3.23,-6.68);
\draw[red,thick]  (6.31,-3.96)-- (10.35,-2.76);
\draw[blue,thick]  (1.59,-3.54)-- (3.23,-6.68);
\node at (1.27,-2.35)  {\textbullet};
\node at (0.86,-0.82)  {\textbullet};
\node at (-4.19,2.72)  {\textbullet};
\node at (1.63,7.06)  {\textbullet};
\node at (10.39,2.38)  {\textbullet};
\node at (10.35,-2.76)  {\textbullet};
\node at (3.23,-6.68)  {\textbullet};
\node at (-6.25,-3.6)  {\textbullet};
\node at (-2.330311907110875,4.106777718752372)  {\textbullet};
\node at (0.40888402202969454,-5.763434893233276)  {\textbullet};
\node[fill=black,regular polygon, regular polygon sides=3,inner sep=1.2pt] at (-1.8533088728854605,2.387980857863238) {};
\node[fill=black,regular polygon, regular polygon sides=3,inner sep=1.2pt] at  (-1.404409502110384,0.7704506243371032) {};
\node[fill=black,regular polygon, regular polygon sides=3,inner sep=1.2pt] at (-0.8063537447789324,-1.384538515317315) {};
\node[fill=black,regular polygon, regular polygon sides=3,inner sep=1.2pt] at (-0.2678470762135259,-3.324952946854639) {};

\node[fill=black,regular polygon, regular polygon sides=4,inner sep=1.7pt] at (-3.33,-2.08) {};
\node at (5.37,0.2)  {\textbullet};
\node at (6.800742234448866,0.8213183408562803)  {\textbullet};
\node at (8.050814882081141,1.3641785742902164)  {\textbullet};
\node at (-0.78,5.26)  {\textbullet};
\node[fill=black,regular polygon, regular polygon sides=4,inner sep=1.7pt] at (9.026564728506182,1.7879105793114491) {};
\node[fill=black,regular polygon, regular polygon sides=4,inner sep=1.7pt] at (2.23,3.12) {};
\node[fill=black,regular polygon, regular polygon sides=4,inner sep=1.7pt] at (1.59,-3.54) {};

\node at (2.0420829064544095,1.164487745291197)  {\textbullet};
\node at (1.811911892190086,-1.2307293718969188)  {\textbullet};

\node[fill=black,regular polygon, regular polygon sides=4,inner sep=1.7pt] at (6.31,-3.96) {};
\node at (5.904567765487212,-2.1657467068370244)  {\textbullet};

    \end{tikzpicture}
    \caption{Illustration of Lemma~\ref{lemma:partitionToSegment}. The green dashed edge and the triangle points are removed at the beginning for the analysis, and added back at the end. The extreme points in $P'$ are represented as square points. The edges in $E'$ are in red. The other edges from $P'$ to the boundary of the convex hull are in blue.}
    \label{fig:partitionToSegment}
\end{figure}

\begin{proof}
The proof is illustrated in Figure~\ref{fig:partitionToSegment}. Let us denote by $G_0=(V_0,E_0)$ the plane graph corresponding to the convex tiling. Observe that some points in $V_0$ might not be in $P$. Also, the relative interior of an edge in $E_0$ might overlap with points in $P$. We assume that $G_0$ is given with a doubly connected edge list (DCEL) structure. If there is an edge between two points on the boundary of the convex hull of $V_0$, but not consecutive, we remove this edge. Note that this decreases the number of faces by $1$, and does not break the convexity property. We denote by $m$ the number of such edges that we have removed. We also remove from $P$ all points contained in the relative interior of an edge between two points on the boundary of the convex hull. We denote by $P''$ the extreme points in $P_i$ that we have not removed. As an edge contains at most two points in $P'$, we have $|P''|\geq |P'|-2m$. Using the DCEL structure, this can be done in $\mathcal{O}(n)$-time. We have obtained a new graph $G=(V,E)$, and there are $f-m$ convex faces in $G$. We denote by $Q$ the set of inner points that are of degree at least $3$ in $G$. We set $k:=|Q|$. Now observe that for each point $p$ in $P''$, there exists at least one edge $e$ in $E$ with one endpoint in $Q$, one endpoint on the boundary of the convex hull, such that $e$ overlaps with a point in $P''$. This is because if we consider $p$ and the two lines going through $p$ and one of the two consecutive vertices in $'P'$ (the one before $p$ and the one after $p$ when going around $P''$ in clockwise order), they define a wedge in which one edge must lie because of convexity. The point in $P''$ can be an endpoint of $e$ or in its relative interior. If for a point $p\in P''$ there are several edges that satisfy the conditions, we choose one arbitrarily. We denote these edges by $E'$. An edge in $E'$ overlaps with exactly one point in $P''$, thus $|E'|=|P''|$. We denote by $E_b$ the edges not in $E'$ that have a point on the boundary of the convex hull and the other in $Q$, and we denote $|E_b|$ by $m'$. The vertices on the boundary of the convex hull are adjacent to two other vertices on the boundary of the convex hull. Moreover, those vertices are incident to $|P''|+m'$ additional edges. We have $2|E|=\sum_{v \in V} \deg(v)\geq 3k+2(n-k)+|P''|+m'=k+2n+|P''|+m'$. By Euler's formula, we have $f-m=|E|-n+1\geq \frac{k+|P''|+m'}{2}+1$.

Now, the solution consists of the union of all edges in $E$ incident to two points in $Q$, with the $m$ edges in $E_0$ that we have removed, and with the $|P''|+m'$ edges in $E' \cup E_b$. We may need those edges as they might overlap with points in $P_i$. Note that there are at most $3k$ edges in $E$ incident to two points in $Q$ as $G$ is plane. Moreover, all points in $P_i$ are indeed covered by the edges in our solution. Thus, we obtain a covering of $P_i$ with $s$ segments, where $s\leq 3k+m+m'+|P''|\leq 3(2(f-m)-|P''|-m')+m+m'+|P''|\leq 6f -5m -2|P''| \leq 6f -5m -2(|P'|-2m)\leq 6f-2|P'|$.
\end{proof}

It is now possible to combine Lemmas~\ref{lemma:segmentToPartition} and~\ref{lemma:partitionToSegment} to prove Theorem~\ref{thm:MCPEquivalentCPNCS}.

\begin{proof}[Proof of Theorem~\ref{thm:MCPEquivalentCPNCS}]
Let us denote by $f'_m$ the minimum number of convex faces in a convex tiling of $P$. We have $f'_m \leq f_m$. First, if $P$ is special, then the three problems are trivial to solve. It remains to prove statement $1$. Recall that we assume that $P_i$ is not empty. As $P$ is special, we need exactly one segment to cover the inner points, and thus $s_m=1$. Now observe that we need between two and four convex faces in a convex partition, and exactly two convex faces in a convex tiling, thus it holds $\frac{s_m}{6}\leq f'_m\leq f_m \leq 8s_m$.

Let us now assume that $P$ is not special. Starting with a covering of the set $P_i$ with $s_m$ non-crossing segments, Lemma~\ref{lemma:segmentToPartition} indicates that it is possible to find a convex partition with at most $4s_m+2|P'|$ convex faces. Therefore we have $4s_m+2|P'|\geq f_m$. Starting from a convex tiling with $f'_m$ convex faces, we know from Lemma~\ref{lemma:partitionToSegment} that there exists a covering of $P_i$ with at most $6f'_m-2|P'|$ non-crossing segments. This implies $6f'_m-2|P'|\geq s_m$. Note that any segment in a covering can cover at most two points in $P'$. Therefore we have $s_m \geq |P'|/2$. Putting everything together, we obtain $\frac{s_m}{6}\leq f'_m\leq f_m\leq 8s_m$.

Let us consider a covering of $P_i$ with $s\leq \lambda s_m$ non-crossing segments. By Lemma~\ref{lemma:segmentToPartition}, we can compute in $\mathcal{O}(n^2)$-time a convex partition of $P$ with at most $f:=4s+2|P'|$ faces. We now have $f=4s+2|P'|\leq 4 \lambda s_m+2|P'|\leq 4 \lambda (6f'_m-2|P'|)+2|P'|\leq 24 \lambda f'_m$. This implies that the convex partition we have is a $24\lambda$-approximation for MCP and for MCT.

Let us consider a convex tiling of $P_i$ with $f\leq \lambda f_m$ convex faces. Note that this encompasses the case where the convex tiling is actually a convex partition. By Lemma~\ref{lemma:partitionToSegment}, we can compute a covering of $P_i$ with at most $s:=6f-2|P'|$ segments. We have $s=6f-2|P'|\leq 6 \lambda f_m- 2|P'|\leq 6 \lambda (4s_m+2|P'|)-2|P'|\leq24 \lambda s_m +10\lambda|P'|\leq 44 \lambda s_m$.
\end{proof}

\section{Approximation algorithms for CPNCS}
\label{sec:approxCPNCS}

In this section we present several approximation algorithms for CPNCS. Let us first consider the ones whose approximation ratio is not output-dependent. The best algorithms in terms of approximation ratio are constant-approximation algorithms. The fastest algorithms take quadratic time. Therefore by 2.~of Theorem~\ref{thm:MCPEquivalentCPNCS}, all the algorithms we present for CPNCS can be used to obtain approximation algorithms for MCP and MCT with the same order of approximation ratio, and the same order of running time. We have also one algorithm for CPNCS which realises an $\mathcal{O}(\log OPT)$-approximation in time $\mathcal{O}(n^8)$, where $OPT$ denotes the minimum number of segments needed. Using 1.~and~2.~of Theorem~\ref{thm:MCPEquivalentCPNCS}, we also derive from it the $\mathcal{O}(\log OPT)$-approximation algorithm for MCP and MCT running in time $\mathcal{O}(n^8)$, where now $OPT$ denotes the minimum number of faces needed in a convex partition, or in a convex tiling, respectively. This is how we prove Theorem~\ref{thm:approxAlgo}.  We first present an easy approximation algorithm running relatively fast, at the cost of a high approximation ratio.

\begin{theorem}
There exists an $\sqrt{n} \log (n)$-approximation algorithm for CPNCS running in $\mathcal{O}(n^2)$-time.
\end{theorem}

\begin{proof}
 Let us denote by $\ell_m$ the minimum number of lines needed to cover $P$, as in the CPL problem. We denote by $s_m$ the minimum number of segments in a valid solution of CPNCS. We have $\ell_m \leq s_m$. Using the greedy algorithm for set cover problems, we can compute a covering of $P$ with $\ell$ lines, where $\ell \leq \log(n)\ell_m$~\cite{johnson1973approximation,lovasz1975ratio}. The greedy algorithm runs in quadratic time. Indeed, it is folklore that the greedy algorithm for covering a set $X$ with the family of subsets $\mathcal{F}\subseteq 2^X$ can be implemented in $\mathcal{O}(Z)$ time, where~$Z=\sum_{F\in \mathcal{F}}|F|$. In our situation, $Z$ is the number of point-line incidences, and so $Z=\mathcal{O}(n^2)$ because each point lies on at most $n-1$ lines. 
 
 We distinguish two cases, depending on the value of $\ell$. We denote by $\phi(n)$ a threshold function, that will be determined later. In the first case, we assume $\ell\geq \phi(n)$. In this situation, we cover each point in $P$ by a segment reduced to that point. We have $n$ segments, and we needed at least $s_m\geq \ell_m \geq \ell / \log(n)\geq \phi(n) / \log(n)$. The approximation ratio is $\frac{n \log(n)}{\phi(n)}$.
 
Now, let us assume $\ell < \phi(n)$. We transform each of the $\ell$ lines into a segment, such that the new segments still cover $P$. Now, at each of the $\mathcal{O}(n^2)$ intersections between the relative interior of a pair of segments, we split one segment into two, such that there is no crossing anymore. Let us denote by $s$ the number of segments obtained. We have $s\leq \ell^2\leq \ell \log (n) \ell_m\leq \phi(n) \log(n) s_m$. The approximation ratio is $\phi(n) \log(n)$.

We make the two approximation ratios equal by setting $\phi(n):=\sqrt{n}$. We obtain a $\sqrt{n} \log(n)$-approximation.
\end{proof}

Mitchell presented in a technical report some approximation algorithms for the problem of covering a point set with a minimum number of pairwise-disjoint triangles~\cite{mitchell1993approximation}. In his problem, the triangles of the covering must be subtriangles of some triangles given as input, for otherwise the problem would be trivial. He makes the assumption that no three points are on a line. We adapt his algorithms to our setting of CPNCS for point sets with no constraint. It seems that there were two mistakes in his proof, that we show how to fix. Let $P$ be a set of $n$ points. By doing a rotation if necessary, we can assume that no two points in $P$ have the same $x$-coordinate. We say that a trapezoid is \emph{constrained} if $1)$ it has two disjoint vertical sides, each lying on a line that contains a point in $P$, and $2)$ the two remaining sides are lying on lines that contain each at least two points in $P$.  Note that there are $\mathcal{O}(n^6)$ constrained trapezoids.

In his paper, Mitchell calls the trapezoids ``canonical'' instead of ``constrained''~\cite{mitchell1993approximation}. We make the choice of changing the name for better clarity later. Also, concerning constraint $1)$, he has the stronger constraint that the vertical sides must each contain a point in $P$. It seems to be a mistake, for otherwise it is not clear how his dynamic programming algorithms work, and some of his arguments do not hold. Anyway, even with his definition, he only uses the fact that there are $\mathcal{O}(n^6)$ constrained trapezoids for computing the running time of his algorithms. Therefore there is no loss in using our definition.

We also allow for some degeneracies. Let us consider a triangle with vertices $a$, $b$ and $c$, not all three on a line. If $a$ is in $P$, the segment with endpoints $b,c$ is vertical and lies on a line that contains a point in $P$, and the segments with endpoints $a,b$ and $a,c$ respectively are contained in some lines $\ell$ and $\ell'$ such that $\ell$ and $\ell'$ contains at least two points in $P$, then we say that the triangle is a constrained trapezoid. If a constrained trapezoid is split into two halves by a vertical line $\ell$ going through its interior, with $\ell$ containing a point in $P$, we obtain two constrained trapezoids. Likewise, if a segment $s$ is in a constrained trapezoid $\tau$, such that $s$ lies on a line that contains at least two points in $P$, $s$ intersects the interior of $\tau$, and the endpoints of $s$ are contained in the vertical sides of $\tau$, then $s$ splits $\tau$ into two constrained trapezoids.

For a set of points $P$ where no two points have the same $x$-coordinate, we define the \emph{enclosing trapezoid} as follows. Let $\ell_1$ be the vertical line that contains the leftmost point in $P$, and let $\ell_2$ be the vertical line that contains the rightmost point in $P$. Let $L$ be the set of all lines containing at least two points in $P$. Observe that no line in $L$ is vertical. We denote by $a$ the highest intersection point between $\ell_1$ and a line in $L$. We denote by $b$ the lowest point intersection point between $\ell_1$ and a line in $L$. Similarly, we denote by $c$ and $d$, respectively, the highest intersection point, respectively the lowest intersection point, between $\ell_2$ and a line in $L$. We denote by $\ell_3$ the line containing $a$ and $c$, and by $\ell_4$ the line containing $b$ and $d$. The \emph{enclosing trapezoid} of $P$ is the constrained trapezoid of $P\cup\{a,b,c,d\}$ defined by $\ell_1$, $\ell_2$, $\ell_3$ and $\ell_4$. It is denoted by $\mathcal{T}_P$.



Mitchell uses in his paper the notion of \emph{guillotine property}. He shows that if there is a covering of the points in $P$ with $s$ elements, then there is a covering of $P$ with at most $\mathcal{O}(s \log s)$ elements having the guillotine property. He then presents an algorithm, and claims that it outputs an optimal solution among all coverings that have the guillotine property. Although we agree that his algorithm outputs a solution with at most $\mathcal{O}(s \log s)$ elements, we present a counterexample to the fact that his algorithm outputs an optimal solution among all coverings that have the guillotine property. Although he considers the problem of covering points with triangle, he reduces the problem to covering a set of points with constrained trapezoids. He defines the strong guillotine property for trapezoids as follows: A set $\mathcal{T}$ of constrained trapezoids has the guillotine property if $a)$ it contains at most one trapezoid, or if $b)$ there exists a partitioning line $\ell$ containing at least two points in $P$ not intersecting the interior of any constrained trapezoid in $\mathcal{T}$, such that the sets of constrained trapezoids on both sides of $\ell$ also have the guillotine property, or if $c)$ there exists a vertical partitioning line $\ell$ not intersecting the interior of any constrained trapezoid in $\mathcal{T}$, such that the sets of constrained trapezoids on both sides of $\ell$ are not empty also have the guillotine property~\cite{mitchell1993approximation}. Mitchell's wording is not exactly the same as ours but the two definitions are equivalent. 

\begin{figure}
    \centering
    \begin{tikzpicture}[scale=0.8]
    
\draw[fill, red, opacity=0.3] (-16,2) -- (-13,8) -- (-13,8.72) -- (-16,9) -- cycle;
\draw[fill, red, opacity=0.3] (-5,4) -- (-14.67,4) -- (-5,4.72) -- cycle;
\draw  (-16,9)-- (-16,2);
\draw  (-16,2)-- (-5,4);
\draw (-16,9)-- (-5,8);
\draw  (-5,8)-- (-5,4);

\node at (-15,8.91) {\textbullet};
\node at (-16,2) {\textbullet};
\node at (-5,4)  {\textbullet};
\node at (-13,8) {\textbullet};
\node at (-14,8.82) {\textbullet};
\node at (-14.67,4) {\textbullet};
\node at (-6,4.65) {\textbullet};

\draw[above] (-15,8.91) node[] {$b$};
\draw[left] (-16,2) node[] {$a$};
\draw[right] (-5,4) node[] {$g$};
\draw[right] (-13,8) node[] {$e$};
\draw[above] (-14,8.82) node[] {$d$};
\draw[above] (-14.67,4) node[] {$c$};
\draw[above] (-6,4.65) node[] {$f$};

\draw[] (-8,8.7) node[] {$\mathcal{T}_1$};
\draw[] (-14.5,7) node[] {$\mathcal{T}_2$};
\draw[] (-8,4.2) node[] {$\mathcal{T}_3$};
    \end{tikzpicture}
    \caption{The constrained trapezoids $\mathcal{T}_2$ and $\mathcal{T}_3$ have the guillotine property, and cover the points in $P$. Mitchell's algorithms applied to $\mathcal{T}_1$ outputs a solution with three trapezoids and is therefore not optimal.}
    \label{fig:counterExample}
\end{figure}

In Figure~\ref{fig:counterExample} we have represented three constrained trapezoids, with $\mathcal{T}_3$ being reduced to a triangle. Observe that $\mathcal{T}_2$ and $\mathcal{T}_3$ have the guillotine property. Indeed the line going through $a$ and $c$ realises condition $b)$ of the guillotine property. Let us now apply Mitchell's algorithm to the constrained trapezoid $\mathcal{T}_1$. In Mitchell's setting, not all constrained trapezoids can be used to cover the points in $P$: they must be subtrapezoids of some given trapezoids. Here the two given trapezoids are $\mathcal{T}_2$ and $\mathcal{T}_3$. Observe that a minimum covering of $P$ uses $\mathcal{T}_2$ and $\mathcal{T}_3$. We claim that Mitchell's algorithm outputs at least three trapezoids. His algorithm recurses on all the ways of splitting $P$ with a vertical line. Observe that any vertical line going through a point in $P$ intersects the interior of $\mathcal{T}_2$ or $\mathcal{T}_3$, and therefore cannot lead to a solution with two trapezoids. In addition, his algorithm recurses on all the ways of splitting $P$ with a segment that contains at least two points in $P$, and whose endpoints are on the vertical sides of $\mathcal{T}_1$. Observe that any such segment $\sigma$ splits the interior of $\mathcal{T}_2$ or $\mathcal{T}_3$. Moreover, the points contained in the trapezoid that is split by $\sigma$ are not on the same side of $\sigma$. This implies that all those recursions will lead to solutions with at least three trapezoids. Figure~\ref{fig:counterExample} thus depicts a counterexample to the fact that Mitchell's algorithm outputs an optimal solution among the ones that have the guillotine property. The reason for that is that the line $\ell$ going through $a$ and $e$, and the line $\ell'$ going through $a$ and $c$, which are the certificates that $\mathcal{T}_2$ and $\mathcal{T}_3$ have the guillotine property, do not intersect $\mathcal{T}_1$ only at its vertical sides. Therefore splitting along $\ell$ and $\ell'$ is not tested by Mitchell's algorithm.

We define the \emph{strong guillotine property} in the special case of segments. We show that if there is a covering of $P$ with $s$ non-crossing segments, then there is a covering of $S$ with $\mathcal{O}(s \log s)$ non-crossing segments having the strong guillotine property. We then present an algorithm that outputs an optimal solution among all the coverings with non-crossing segments having the strong guillotine property. Let $S$ be a set of non-crossing segments covering $P$. We assume that the endpoints of the segments in $S$ are in $P$. We say that $S$ has the strong guillotine property with respect to a constrained trapezoid $\mathcal{T}$ that contains all segments in $S$ if $a)$ $S$ contains at most one segment, or if $b)$ there exists a partitioning line $\ell$ containing at least two points in $P$ and at least one segment in $S$, such that for any segment $s\in S$, $\ell$ either contains $s$ or does not intersect the relative interior of $s$, and $\ell$ splits $\mathcal{T}$ into two constrained trapezoids $\mathcal{T}_1$ and $\mathcal{T}_2$, such that the segments in $\mathcal{T}_1$, respectively $\mathcal{T}_2$, have the strong guillotine property with respect to $\mathcal{T}_1$, respectively $\mathcal{T}_2$, or if $c)$ there exists a vertical line not intersecting with the relative interior of any segment in $S$, that splits $\mathcal{T}$ into two constrained trapezoids $\mathcal{T}_1$ and $\mathcal{T}_2$, such that the segments in $\mathcal{T}_1$, respectively $\mathcal{T}_2$, have the strong guillotine property with respect to $\mathcal{T}_1$, respectively $\mathcal{T}_2$. Observe that the line $\ell$ in case $b)$ only intersects the vertical sides of $\mathcal{T}$, for otherwise $\ell$ would not split $\mathcal{T}$ into constrained trapezoids. We simply say that $S$ has the strong guillotine property if it has the strong guillotine property with respect to the enclosing trapezoid $\mathcal{T}_P$.



\begin{lemma}\label{lemma:coverBSP}
If there exists a covering of $P$ with $s$ non-crossing segments, then there exists a covering of $P$ with $\mathcal{O}(s \log(s))$ non-crossing segments with the strong guillotine property.
\end{lemma}

\begin{proof}
Recall that we assume that the endpoints of the segments are in $P$, by cropping them if need be. We can even crop some segments further such that they are pairwise-disjoint (it may be that now some segments are reduced to points). Consider the endpoints of the segments in that covering, that we denote by $P'$. We denote $|P'|$ by $n$, and we have $n\leq 2s$. Note that no two points in $P'$ have the same $x$-coordinate. We denote by $X$ the set of $x$-coordinates of the points in $P'$. We now consider the segment tree based on $X$, as defined in~\cite{preparata1985computational}. The segment tree defines some canonical intervals. Each interval whose endpoints are in $X$ is partitioned into $\mathcal{O}(\log s)$ canonical intervals. We partition each segment in the covering, such that the projection on the $x$-axis of each new segment is a canonical interval. Therefore we obtain a covering of $P$ with $\mathcal{O}(s \log(s))$ non-crossing segments. We claim that this family of segments has the strong guillotine property. Let us denote by $x_i$, $1\leq i \leq n$ the elements in $X$, ordered by increasing value. We distinguish two cases. If there exists a segment $\sigma$ whose projection on the $x$-axis is equal to the interval $[x_1,x_n]$, then we recurse on the parts above and below $\sigma$ which contain some segments. Observe that if $n=2$ we are done. If there is no such segment, then by definition of a segment tree, there is no segment in the covering whose relative interior intersects the vertical line $\ell$ with $x$-coordinate equal to $x_{\lfloor(1+n)/2 \rfloor}$. Thus we can recurse on the left and right side of $\ell$.
\end{proof}

\begin{theorem}\label{thm:logOPTapprox}
There exists an $\mathcal{O}(\log(OPT))$-approximation algorithm running in $\mathcal{O}(n^8)$-time for CPNCS.
\end{theorem}

\begin{proof}
We explain how to recursively compute a minimum covering of $P$ with non-crossing segments under the constraint that the solution has the strong guillotine property. The approximation ratio for the CPNCS problem when this additional constraint is removed follows from Lemma~\ref{lemma:coverBSP}. If $P$ is empty, we return no segment, which is a valid solution. If $P$ can be covered with a single segment, we return that segment. This can be tested in $\mathcal{O}(n^2)$ time using duality. Now let us assume that not all points in $P$ are on a line. We compute the enclosing trapezoid $\mathcal{T}_P$ of $P$. We consider the four vertices $a,b,c,d$ of $\mathcal{T_P}$. We start by adding the segment with endpoints $a,c$, and the segment with endpoints $b,d$. Now all the points to cover are within the enclosing trapezoid $\mathcal{T}_P$. We distinguish two cases, according to whether a segment with endpoints on the vertical sides of $\mathcal{T}_P$ is in a minimum covering with non-crossing segments having the strong guillotine property. If it is, we can add it to the solution and recurse on the two new constrained trapezoids. If no such segment is part of a minimum solution, then there exists a vertical line $\ell$ that splits a minimum solution into two parts, such that $\ell$ does not intersect the relative interior of any segment in that minimum solution. We can recurse on the $\mathcal{O}(n)$ choices of splitting vertically the constrained trapezoid into two constrained trapezoids. For each of the $\mathcal{O}(n^2)$ recursions, we compute the number of segments corresponding to that solution, and we output the solution corresponding to the one that minimises the number of segments.

To optimise we can do dynamic programming, and solve first the thinnest constrained trapezoids (in terms of width on the $x$-axis). There are $\mathcal{O}(n^6)$ constrained trapezoids, and we take quadratic time for each of them, so the total running time is $\mathcal{O}(n^8)$.
\end{proof}

It is possible to use the segment tree technique for the computation as done by Mitchell. It reduces the running time to $\mathcal{O}(n^7)$ at the cost of a slightly worse approximation ratio.

\begin{theorem}
There exists an $\mathcal{O}(\log(n))$-approximation algorithm running in $\mathcal{O}(n^7)$-time for CPNCS.
\end{theorem}

\begin{proof}
We consider the set $X$ of $x$-coordinates of points in $P$. We compute the corresponding segment tree in $\mathcal{O}(n \log (n))$-time. Let us consider a minimum covering of $P$ with $s$ non-crossing segments. We crop the segments so that their endpoints are in $P$. We partition each segment in the covering such that the projection of each new segment on the $x$-axis is a canonical interval of the segment tree. We say that a segment is \emph{canonical} if its projection on the $x$-axis is a canonical interval. We observe that there is a covering of $P$ with $\mathcal{O}(s \log(n))$ non-crossing canonical segments. Thus, we can adapt the algorithm of Theorem~\ref{thm:logOPTapprox} to obtain an $\mathcal{O}(\log(n))$-approximation, by outputing a minimum covering of $P$ with non-crossing canonical segments.

We call a constrained trapezoid whose projection on the $x$-axis is a canonical interval a \emph{canonical trapezoid} (note that this definition is not the same as Mitchell's). As there are $\mathcal{O}(n)$ canonical intervals, there are $\mathcal{O}(n^5)$ canonical trapezoids. We do as in the algorithm of Theorem~\ref{thm:logOPTapprox}. The difference is that when we assume that there exists a vertical line $\ell$ that splits $P$, with $\ell$ not intersecting the relative interior of any segment in an optimal solution, we can assume that the $x$-coordinate of $\ell$ is equal to the median of $X$. For each canonical trapezoid we still do $\mathcal{O}(n^2)$ recursions, so the overall running time of the algorithm is $\mathcal{O}(n^7)$.
\end{proof}


We say that a point set $P$ is \emph{$k$-directed} if there exists a set $D$ of $k$ directions, such that for any line $\ell$ that contains at least three points in $P$, the direction of $\ell$ is in $D$. Without loss of generality, we can assume that the two directions of a $2$-directed point set $P$ are vertical and horizontal. Indeed, this has no impact for the CPNCS problem. For convenience, for any set of directions $D$ and any segment $s$ reduced to a point, we say that the direction of $s$ is in $D$. We say that a set of segments $S$ has the \emph{autopartition property} if $|S|=1$, or if there exists a line $\ell$ which contains at least one segment in $S$, and splits $S$ into two sets that are either empty or have the autopartition property. The relative interior of a segment in $S$ is either contained in $\ell$ or does not intersect $\ell$.

\begin{lemma}\label{lemma:kDirected}
Let $P$ be a $k$-directed point set with set of directions $D$. If there exists a covering of $P$ with $s$ non-crossing segments, then there exists a covering of $P$ with $\mathcal{O}(sk)$ non-crossing segments having the autopartition property, such that the direction of each segment is in $D$. If $k=2$ then there exists a covering with at most $4s$ non-crossing segments, being vertical or horizontal, having the autopartition property.
\end{lemma}

\begin{proof}
Let $D$ be the set of $k$ directions, such that the direction of any line that contains at least three points in $P$ belongs to $D$. From the covering with $s$ segments, we can obtain a covering with at most $2s$ segments such that the direction of each segment is in $D$. Indeed, a segment in the covering whose direction is not in $D$ contains at most two points in $P$. We crop some segments if necessary such that no two segments intersect, and they still cover $P$. Now, we use a theorem by T{\'o}th who showed that any set of $s'$ disjoint segments having up to $k$ directions have an autopartition of size $\mathcal{O}(s'k)$~\cite{toth2003binary}. This immediately implies the result. 

Let us now assume $k=2$. Let us consider a set of $s'$ segments that are vertical or horizontal. There exists a partition of the segments that contains at most $2s'$ segments, and which has the autopartition property~\cite{dumitrescu2004binary}. An upper bound of $3s'$ was first shown by Paterson and Yao~\cite{paterson1992optimal}. It was then improved to $2s'$ by d'Amore and Franciosa~\cite{d1992optimal}, although not explicitly. Dumitrescu, Mitchell and Sharir made the result explicit later~\cite{dumitrescu2004binary}. 
\end{proof}

\begin{theorem}\label{thm:kDirected}
There exists an $\mathcal{O}(k)$-approximation algorithm for CPNCS in $k$-directed sets running in $n^{\mathcal{O}(k)}$. Furthermore, there exists a $4$-approximation algorithm for CPNCS in $2$-directed sets running in time $\mathcal{O}(n^5)$.
\end{theorem}

\begin{proof}
Let $P$ be a $k$-directed point set with set of directions $D$. We show how to compute an optimal covering of $P$ with non-crossing segments, such that the solution has the autopartition property, and the direction of each segment in the solution is in $D$. Lemma~\ref{lemma:kDirected} immediately implies that our algorithm realises an $\mathcal{O}(k)$-approximation, and a $4$-approximation for the special case $k=2$.

The recursion of the algorithm is as follows. If there exists a segment that contains all points in $P$, we add it to the solution. This can be tested in time $\mathcal{O}(n^2)$. Otherwise, for each direction $\delta$ in $D$, we recurse in the $\mathcal{O}(n)$ ways of splitting $P$ with a line $\ell$, such that the direction of $\ell$ is $\delta$, and $\ell$ contains at least a point in $P$. We add to the solution the shortest segment containing the points in $\ell$, and we recurse on both sides of $\ell$.

The subsets of $\mathbb{R}^2$ we are considering in the recursion are defined by giving for each direction the two extreme points for that direction. For each point we have $\mathcal{O}(n)$ choices, so in total there are $n^{2k}$ of such subsets. For each subset by using dynamic programming, we need $\mathcal{O}(nk)$-time. Thus, the total running time is $n^{\mathcal{O}(k)}$, and simply $\mathcal{O}(n^5)$ for the special case $k=2$.
\end{proof}

A natural question is whether we can use the autopartition property when the number of directions is not fixed. It is known that a set of $s$ pairwise-disjoint segments allows for an autopartition with $\mathcal{O}(s \log s / \log \log s)$ segments~\cite{toth2011binary}. This is tight~\cite{t2001note}. Using the techniques we have presented, one could hope to obtain an $\mathcal{O}(\log OPT / \log \log OPT)$-approximation algorithm. However, this autopartition might not have any good structure, and so we cannot use dynamic programming because there are too many subsets of $\mathbb{R}^2$ to consider. In any case, because of the tightness on the number of segments in the autopartition, it seems that the autopartition technique cannot be used to obtain constant-approximation algorithms.

\section{Fixed-parameter algorithm for CPNCS}
\label{sec:FPT}
As mentioned in the introduction, there are known fixed-parameter algorithms for MCP, where the parameter is the number of inner points. We present here a fixed-parameter constant-approximation algorithm for MCP and MCT, where the parameter is the number of faces in a minimum convex partition or a minimum convex tiling, respectively. For point sets where no three points are on a line, the minimum number of convex faces is at least half the number of inner points~\cite{knauer2006approximation}, and the number of convex tiles is at list a sixth of the number of inner points~\cite{dumitrescu2012minimum}. However, as shown in Figure~\ref{fig:noLowerBound}, when we allow for several points to be on a line, the number of inner points can be arbitrarily larger than the number of convex faces in a minimum convex partition. If the number of inner points is significantly higher than the number of convex faces needed, our algorithm has a lower running time. We first show that CPNCS is in FPT.

\begin{theorem}\label{thm:FPTcoveringSegments}
It is possible to compute in time $\mathcal{O}(2^{k^2}k^{7k}+n^4 \log n)$ whether a point set $P$ can be covered with at most $k$ non-crossing segments, and to output such a covering if it exists. 
\end{theorem}

The proof uses a kernelisation technique presented by Langerman and Morin for CPL~\cite{langerman2005covering}. Assume there is a line $\ell$ that contains at least $k+1$ points in $P$. Then in any covering of $P$ with at most $k$ lines, $\ell$ must be in the covering. Otherwise, we would need at least $k+1$ lines to cover the points contained in $\ell$. Now one can compute all of these lines that contain at least $k+1$ points, dismiss all of the covered points, until no line covers more than $k$ of the remaining points. If there remains more than $k^2$ points, then there is no covering of the point set with at most $k$ lines. Otherwise, one can compute every way of covering the $\mathcal{O}(k^2)$ remaining points, and check whether there is one that uses in total at most $k$ lines. In our setting, we are looking for a covering with non-crossing segments, which makes it more difficult. Indeed, if a line $\ell$ contains at least $k+1$ points, we only know that $\ell$ must contain at least one segment of the covering. This means that we cannot simply dismiss the points covered by such a line. Also, we have to be careful about crossings. Before proving Theorem~\ref{thm:FPTcoveringSegments}, we first show several lemmas. For a point set $P$, we say that a segment $s$ is a \emph{$P$-segment} if its endpoints are in $P$. Recall that we only consider coverings of a point set $P$ with non-crossing $P$-segments.

\begin{definition}\normalfont
Let $P$ be a point set, and let $s$ and $t$ be two crossing $P$-segments. We denote by $p$ the intersection of $s$ and $t$. We determine four points in $P$, that we call the \emph{points enclosing $p$}. There are two points on $s\cap P$ and two points on $t\cap P$. The two points on $s\cap P$, denoted by $u$ and $v$, are such that the segment with endpoints $u$ and $v$, which we denote by $uv$, is the shortest $P$-segment contained in $s$ whose relative interior contains $p$. Likewise, the two points $u'$ and $v'$ are such that $u'v'$ is the shortest $P$-segment contained in $t$ whose relative interior contains $p$. The points $u$, $v$, $u'$ and $v'$ are the points enclosing $p$.
\end{definition}

\begin{lemma}\label{lemma:FPTpreprocess}
Given a set $P$ of $n$ points, it is possible to compute in time $\mathcal{O}(n^4 \log n)$ the pairs of crossing $P$-segments, to find whether their intersection $p$ is in $P$, and to store the points enclosing $p$. Additionally, we can also store for each $P$-segment how many points in $P$ they contain, and the list of those points.
\end{lemma}

\begin{proof}
Let $s$ be a $P$-segment. We first store the number of points contained in $s$, as well as the list consisting of those points. We then sort the list so that when going from one endpoint of $s$ to the other, the points appear consecutively on the list. As there are $\mathcal{O}(n^2)$ of such segments, this preprocessing can be done in time $\mathcal{O}(n^3\log n)$. Let $s$ and $s'$ be some crossing $P$-segments. There are $\mathcal{O}(n^4)$ pairs of such segments. We denote by $p$ the intersection of $s$ and $s'$. We check whether $p$ is in $P$. This can be done in time $\mathcal{O}(\log n)$ by searching through the list of points in $s$. We denote by $u$, $v$, $u'$ and $v'$ the points enclosing $p$. Observe that given $s$ and $s'$, it takes time $\mathcal{O}(\log n)$ to find the points enclosing $p$, and to test whether $p$ is in $P$. Thus when considering all pairs of segments, we can compute this information in time $\mathcal{O}(n^4 \log n)$, and so given the endpoints of some crossing $P$-segments $s$ and $s'$, we can retrieve this information in constant time. Thus, the total running time of the algorithm is in $\mathcal{O}(n^4 \log n)$.
\end{proof}

\begin{lemma}\label{lemma:FPTcandidate}
Given a set $P$ of $n$ points, and a natural number $k$, it is possible to find in time $\mathcal{O}(2^{k^2}+n^4 \log n)$ either a certificate that there is no covering of $P$ with at most $k$ non-crossing segments, or to output a family $\mathcal{F}$ of $\mathcal{O}(2^{k^2})$ sets $S$ containing at most $k$ non-crossing $P$-segments, with the following properties: For any fixed covering of $P$ with at most $k$ non-crossing $P$-segments, there exists a set $S$ in $\mathcal{F}$ such that a) any segment $s\in S$ contains at least $k+1$ points in $P$, b) for each segment $t$ of the covering, if $|P\cap t \cap s|\geq 2$ for some $s\in S$, then $t$ is contained in $s$, and c) if a segment of the covering contains at least $k+1$ points in $P$, then it is contained in a segment in $S$.
\end{lemma}

Let $P$ be a point set and let $k$ be a natural number. Observe that if a set $S$ of segments satisfies property a), then in a covering with at most $k$ segments of $P$, each segment $s$ in $S$ contains at least one segment $t$ of the covering, such that $|P\cap t|\geq 2$. Indeed if there exists a segment $s\in S$ such that for any segment $t$ in the covering, we have that $s \cap t$ contains at most one point in $P$, then at least $k+1$ segments are needed to cover the points in $P\cap s$. This implies that if $S$ consists of $m$ segments and satisfies properties a) and b), then there are at least $m$ segments in the considered covering of $P$ with non-crossing segments.

\begin{proof}[Proof of Lemma~\ref{lemma:FPTcandidate}]
We first do some preprocessing by using the algorithm of Lemma~\ref{lemma:FPTpreprocess}. This takes $\mathcal{O}(n^4 \log n)$ time. We create a list $L$ of segments, which at the beginning is empty, and will contain the segments in $S$ when we are done. For each line $\ell$ that contains at least $k+1$ points, we find the extremal points $p$ and $q$ of $P$ contained in $\ell$ in time $\mathcal{O}(n)$. Then we add the line segment with endpoints $p$ and $q$ to $L$. Using the algorithm presented by Guibas {\em et al.}~\cite{guibas1996exact}, we can compute all lines containing more than $k$ points in time $\mathcal{O}(\frac{n^2}{k}\log(\frac{n}{k}))$. If there are more than $k$ of such lines, we already know that there is no covering of $P$ with at most $k$ non-crossing segments of $P$. Indeed such a covering can only exist if there exists a covering of $P$ with at most $k$ lines. Let us now assume that there are at most $k$ such lines. We add all corresponding segments to $L$ in total time $\mathcal{O}(kn+\frac{n^2}{k}\log(\frac{n}{k}))$. Let us show that the segments in $L$ satisfy properties a), b) and c), although they might still be crossing. First, property a) holds by definition. Moreover property b) holds for all covering of $P$ with at most $k$ segments because a segment in $L$ containing points $p$ and $q$ also contains all points on the line $(p,q)$. Finally, property c) also holds trivially for all covering of $P$ with at most $k$ segments. 

We are now going to modify $L$ and make copies of it while maintaining the fact that properties a), b) and c) hold. Our aim is that no two segments in $L$ cross. Let us consider one segment $s$ in $L$ which is crossed by another segment $s'$ in $L$. We denote by $p$ the intersection of $s$ and $s'$. We retrieve the points $u$ and $v$ such that $uv$ is the shortest $P$-segment in $s$ whose relative interior contains $p$. We do likewise with $u'$ and $v'$ in $s'$. Observe that not both $uv$ and $u'v'$ can be in a covering of $P$ with non-crossing segments. More generally, in a valid covering, at least one of $uv$ and $u'v'$ is not contained in any segment of the covering. We create one copy of $L$, and recurse on two cases, one where we assume that $uv$ is not contained in a segment of the covering, and one where we assume that $u'v'$ is not contained in a segment of the covering. In the case where we assume that $uv$ is not in the covering, $s'$ remains in $L$, and might still be removed at a later step. Let us assume for now that $uv$ is not contained in a segment of the covering. We remove $s$ from $L$. The candidate segment $s'$ splits $s$ at $p$ into two sides. Let us denote by $x$ and $y$ the endpoints of $s$, with $u$ being closer to $x$ than $v$ is. If $p$ is not in $P$, we consider the segments $xu$ and $vy$. If $p$ is in $P$, we consider the segments $xp$ and $py$. Any of the two new segments that contains more than $k$ points in $P$ is added to $L$. Indeed property a) holds by definition. Moreover property b) holds because $s$ was in $L$, and we are assuming that the segment $uv$ is not contained in a segment of the covering. If a segment contains at most $k$ points, we do not add it to $L$. We claim that property c) still holds. This is because if a point $q\in P$ which lies on a line that contains more than $k$ points is not contained in some segment in $L$, that means that if a segment $t$ contains $q$ as well as at least $k$ other points in $P$, then $t$ also contains some segment which we are assuming not to be contained in a covering. 

If we obtain more than $k$ segments in $L$, we stop this branch of the recursion, as we already know that there is no valid covering of $P$ with at most $k$ segments, assuming that $uv$ is not contained in a segment of the covering. We now iterate over all crossing segments in $L$. We obtain $O(k)$ segments in $L$, which are by construction non-crossing. As we stop after $\mathcal{O}(k^2)$ recursions, the number of leaves of the recursion we consider is in $\mathcal{O}(2^{k^2})$. We would like to say that each recursion implies the existence of one more segment in a covering with non-crossing segments, but this is a priori not the case. Therefore, if the number of lines containing more than $k$ points is in $\Omega(k)$, we might have to do $\Omega(k^2)$ recursions. We can do the computation in total time $\mathcal{O}(2^{k^2}+kn+\frac{n^2}{k}\log(\frac{n}{k}))$, using the information we preprocessed. If we add to it the running time of the preprocessing, the total running time of the algorithm is in $\mathcal{O}(2^{k^2}+n^4 \log n)$.
\end{proof}

\begin{proof}[Proof of Theorem~\ref{thm:FPTcoveringSegments}]
We first use the algorithm of Lemma~\ref{lemma:FPTcandidate}. In particular, we keep the information that is preprocessed with the algorithm of Lemma~\ref{lemma:FPTpreprocess}. If we have a certificate that there is no covering of $P$ with at most $k$ non-crossing segments, we stop. Let us assume that the algorithm outputs a family $\mathcal{F}$ of $\mathcal{O}(2^{k^2})$ lists $L$, such that $\mathcal{F}$ satisfies the conditions of Lemma~\ref{lemma:FPTpreprocess}. Let us consider a fix covering of $P$ with at most $k$ non-crossing segments, assuming one exists. We guess the corresponding list of segments $L$ in time $\mathcal{O}(2^{k^2})$. We call the segments in $L$ \emph{candidate segments}.

Let us denote by $m$ the number of candidate segments, and by $m'$ the number of points not contained in a candidate segment. Computing $m$ and $m'$ takes $\mathcal{O}(n)$ time. If $m'$ is larger than $k^2$, we output that there is no solution. Indeed, by property c), a segment in the covering can contain at most $k$ points which are not contained in some candidate segment. If $m+m'$ is at most $k$, then we take the covering consisting of all candidate segments, and segments reduced to a point for each point which is currently uncovered. Otherwise, for a covering to have at most $k$ non-crossing segments, there must be a segment that contains at least two points which are currently not covered by candidate segments. Indeed by properties a) and b) we know that for each candidate segment $s$, there exists a segment $t$ in the covering that is contained in $s$, and therefore $t$ contains no point currently uncovered. Thus to have a covering with fewer than $m+m'$ segments, there must be a segment $\sigma$ in the covering which contains at least two points for now uncovered. Recall that $m'$ is at most $k^2$. We consider the $\mathcal{O}(k^4)$ lines going through at least two uncovered points. As we have shown, there exists a line $\ell$ that contains two uncovered points, and also contains $\sigma$. Observe that the endpoints of $\sigma$ might be contained in some candidate segments. We first guess in $\mathcal{O}(k^4)$ time the largest segment $\tau$ contained in $\sigma$, such that the endpoints of $\tau$ are uncovered points. If $\ell$ intersects a candidate segment at a point $p \in P$, such that $p$ is not in $\tau$, we want to guess whether $p$ is an endpoint of $\sigma$. As there are at most $k$ candidate segments, given $\tau$, we can guess the endpoints of $\sigma$ in $\mathcal{O}(k^2)$ time. Therefore, we can guess $\sigma$ in $\mathcal{O}(k^6)$ time. We find in $\mathcal{O}(k)$ time the list of candidate segments $\sigma$ intersects, and also check that $\sigma$ does not intersect any segment that we have already taken in the solution during a past iteration. For each candidate segment $s$ that $\sigma$ intersect, we do as in the algorithm of Lemma~\ref{lemma:FPTpreprocess} and we split $s$ into two sides. We also update $m$ and $m'$. For a specific candidate segment, this can be done in constant time thanks to the preprocessing. We now iterate from the beginning of the paragraph. At each iteration, we are either done, or we have one more segment in our partial covering. Therefore we iterate at most $k$ times. The total running time of this algorithm (not including the preprocessing) is in $\mathcal{O}(k^{7k})$. The total running time of the algorithm is in $\mathcal{O}(2^{k^2}k^{7k}+n^4 \log n)$.
\end{proof}

\begin{theorem}
It is possible to compute in time $\mathcal{O}(2^{36f^2}f^{42f+1}+n^4 \log n)$ a convex partition of a point set $P$ with at most $24f$ convex faces, where $f$ denotes the minimum number of convex faces required. The same holds when considering convex tilings.
\end{theorem}

\begin{proof}
We first compute a minimum covering of the inner points in time $\mathcal{O}(2^{s^2}s^{7s+1}+n^4 \log n)$ by applying the algorithm of Theorem~\ref{thm:FPTcoveringSegments} for $k=1,2, \dots, s$, where $s$ denotes the minimum number of segments required in a covering of the inner points. Then, by 2. of Theorem~\ref{thm:MCPEquivalentCPNCS}, we obtain in $\mathcal{O}(n^2)$-time a convex partition with at most $24f$ convex faces. The same holds with convex tilings for the same arguments. As by 1. of Theorem~\ref{thm:MCPEquivalentCPNCS}, we have $s\leq 6f$, the total running time of the algorithm is as stated. 
\end{proof}

There is a strong similarity between CPNCS and Maximum Independent Set in Segment Intersection Graphs~(MISSIG). As an example, we show in Section~\ref{sec:NPhardCPNCS} that CPNCS is NP-hard by doing a reduction from MISSIG. We have shown that CPNCS is in FPT, but Marx has shown that MISSIG is W[1]-hard~\cite{marx2006parameterized}. We do a sanity check and explain why his hardness reduction does not apply to CPNCS. In his reduction, there are $f(k)$ gadgets, each gadget containing $\mathcal{O}(n^2)$ segments. In each gadget, a constant number of segments has to be taken in an independent set of size $k$. We could try to mimic our NP-hardness reduction of Section~\ref{sec:NPhardCPNCS}, take the same set of segments as Marx, and then replace each segment by a set of four collinear points. Taking a segment in the independent set corresponds to covering these four points with one segment. If a segment is not taken in the independent set, then we need two segments to cover the four points. Therefore, one needs $\Omega(n)$ segments in each gadget to cover the points, and not some constant number, which implies that the W[1]-hardness reduction we are trying to do is not valid.

We now discuss why the FPT algorithm and the techniques presented for CPNCS do not give us an FPT algorithm for MCP where the parameter is the number of faces. One can first use Lemma~\ref{lemma:FPTcandidate} to guess some candidate segments. Then it is possible to guess how many vertices of degree at least $3$ lie on each of the $\mathcal{O}(k)$ candidate segments. Then we can enumerate all plane graphs on this vertex set, and guess which one corresponds to our convex partition. However, for now we have only guessed on which candidate segment does a vertex lie. It remains to check whether all those points can be placed at points in $P$, while preserving the fact that the graph is a convex partition. This can be modelled as an integer linear programming problem, but as the number of constraints is linear in $n$, this does not give an FPT algorithm.

\section{NP-hardness of MCP and MCT}
\label{sec:NPMCP}

Our proof builds upon gadgets introduced by Lingas~\cite{lingas1982power}. He used them to prove NP-hardness of several decision problems, including \emph{Minimum Convex Partition for Polygons with Holes} and \emph{Minimum Rectangular Partition for Rectilinear Polygons with Holes}. In the second problem, Steiner points are allowed. 
However, as noted by Keil~\cite{keil1985decomposing}, one can easily adapt Lingas' proof to not use Steiner points. We use a similar idea to prove NP-hardness of the MCP problem. 
Lingas' proofs for the two problems are similar, and consist in a reduction from the following variation of planar $3$-SAT. The instances are a CNF formula $F$ with set of variables $X$ and set of clauses $C$, and a planar bipartite graph $G=(X \cup C,E)$, such that there is an edge between a variable $x\in X$ and a clause $c \in C$ if and only if $x$ or $\bar{x}$ is a literal of $c$. Moreover, each clause contains either two literals or three, and if it contains three, the clause must contain at least one positive and one negative literal. Lingas refers to this decision problem as the \emph{Modified Planar $3$-SAT} (MPLSAT). Lingas claims that planar satisfiability can easily be reduced to MPLSAT by adding new variables. As planar satisfiability was shown to be NP-complete by Lichtenstein~\cite{lichtenstein1982planar}, this would imply that MPLSAT is NP-complete too. For the sake of completeness, we remark that it is not clear why adding these new variables would not break the planarity of the graph. This can be solved by considering the following lemma of Lichtenstein:

\begin{lemma}[Lichtenstein~\cite{lichtenstein1982planar}]
Planar satisfiability is still NP-complete even when, at every variable node, all the arcs representing positive instances of the variable are incident to one side of the node and all the arcs representing negative instances are incident to the other side.
\end{lemma}

This lemma can easily be strengthened to the case of planar $3$-SAT as noted by Lichtenstein, and explicitly done by Tippenhauer~\cite{tippenhauer2016planar} in his Master's thesis. From here a reduction to MPLSAT becomes indeed straightforward.

\begin{theorem}\label{NPC}
MPLSAT can be reduced in polynomial time to MCP, and to MCT.
\end{theorem}

As it is easy to see that MCP is in NP, Theorem~\ref{NPC} implies that MCP is NP-complete. The question whether MCT is in NP is still open. Let $P$ be a point set, and let us consider the set $\mathcal{L}$ of all lines going through at least two points in $P$. Let $P'$ be the set of points at the intersection of at least two lines in $\mathcal{L}$. One might think that there exists a minimum convex tiling, such that all Steiner points belong to $P'$. We show in Figure~\ref{fig:maybeNotInNP} that this is not the case.

\begin{figure}[ht]
    \centering
    \begin{tikzpicture}[scale=1]

\draw[fill, yellow, opacity=0.5] (-2,0)--(-0.5,1.5)--(-2,3)--cycle;
\draw[fill, blue, opacity=0.3] (-2,0)--(-0.5,1.5)--(2,1.5)--(2,0)--cycle;
\draw[fill, gray, opacity=0.3] (-2,3)--(-0.5,1.5)--(2,1.5)--(2,3)--cycle;

\node at (-2,0) {\textbullet};
\node at (-2,3) {\textbullet};
\node at (2,0) {\textbullet};
\node at (2,3) {\textbullet};

\node at  (-1.7267867970742536,2.7267867970742534) {\textbullet};

\node at  (-1.5,2.5) {\textbullet};

\node at  (-1.2072152171539057,2.2072152171539057) {\textbullet};

\node at  (-1,2) {\textbullet};

\node at  (-0.7636785025877549,1.7636785025877548) {\textbullet};

\node at  (-1.9978097731153714,-0.0021870334815434767) {\textbullet};

\node at  (-1.7249955027911472,0.2714245203924688) {\textbullet};

\node at  (-1.498539848743581,0.498541977678971) {\textbullet};

\node at  (-1.2061825759660196,0.7917536472737127) {\textbullet};

\node at  (-0.9992699243717904,0.9992709888394855) {\textbullet};

\node at  (-0.7632934920909107,1.2359370482698966) {\textbullet};

\node[red] at  (-0.5,1.5) {\textbullet};
\node[red] at  (2,1.5) {\textbullet};
\node at  (0.5,1.5) {\textbullet};

    \end{tikzpicture}
    \caption{Points in black are the input point set $P$. The points in red are Steiner points in the unique minimum convex tiling of $P$. The point in red to the right does not belong to the set $P'$.}
    \label{fig:maybeNotInNP}
\end{figure}

We first prove Theorem~\ref{NPC} for the MCT problem. We do the reduction from MPLSAT to MCT by constructing a point set in three steps. First we construct a non-simple polygon, in a similar way as in Lingas' proof, with some more constraints. Secondly, we add some line segments to build a grid around the polygon, and finally we discretise each line segment into evenly spaced collinear points. The idea of the first part is to mimic Lingas' proof. The second part makes the correctness proof easier, and the last part transforms the construction into our setting. The aim of the grid is to force the polygons in a minimum convex tiling to be rectangular.

We use the gadgets introduced by Lingas, namely cranked wires and junctions~\cite{lingas1982power}. A wire is shown in Figure~\ref{fig:wire}. It consists of a loop delimited by two polygons, one inside the other. In Lingas' construction, the two polygons are simple, and a wire is therefore a polygon with one hole. Moreover in his proof the dimensions of the cranks do not matter. In our case, the polygon inside is not simple, and each line segment has unit length. Each wire is bent several times with an angle of $90^\circ$, as shown in Figure~\ref{fig:wire}, in order to close the loop.

\begin{figure}[ht]
    \centering
    \begin{tikzpicture}[scale=0.6]

\draw[fill, gray, opacity=0.3] (0,-2)--(0,-1)--(-1,-1)--(-1,0)--(0,0)--(0,1)--(-1,1)--(-1,2)--(0,2)--(0,3)--(1,3)--(1,4)--(2,4)--(2,3)--(3,3)--(3,4)--(4,4)--(4,3)--(4,2)--(3,2)--(3,1)--(2,1)--(2,2)--(1,2)--(1,1)--(2,1)--(2,0)--(1,0)--(1,-1)--(2,-1)--(2,-2)--(1,-2)--cycle;

\draw[thick] (0,-2)--(0,-1)--(-1,-1)--(-1,0)--(0,0)--(0,1)--(-1,1)--(-1,2)--(0,2)--(0,3)--(1,3)--(1,4)--(2,4)--(2,3)--(3,3)--(3,4)--(4,4)--(4,3);
\draw[thick] (1,-2)--(2,-2)--(2,-1)--(1,-1)--(1,0)--(2,0)--(2,1)--(1,1)--(1,2)--(2,2)--(2,1)--(3,1)--(3,2)--(4,2);

\draw[red] (0.5,-5)--(0.5,2.5)--(6,2.5)--(6,6)--(11,6)--(11,-1)--(7,-1)--(7,-5)--cycle; 

    \end{tikzpicture}
    \caption{A cranked wire, edges are in black and its interior is in grey. The wire follows the whole red loop, but for the sake of simplicity, only a section of the wire at a bend is drawn.}
    \label{fig:wire}
\end{figure}

The wires are used to encode the values of the variables, with one wire for each variable. We are interested in two possible tilings of a wire, called vertical and horizontal, which are shown in Figure~\ref{fig:wirePartition}. 

\begin{figure}[ht]
    \centering
    \begin{tikzpicture}[scale=0.7]

\draw[fill, blue, opacity=0.3] (0,-2)--(0,3)--(1,3)--(1,-2)--cycle;
\draw[fill, blue, opacity=0.3] (2,3)--(2,1)--(3,1)--(3,3)--cycle;
\draw[fill, red, opacity=0.3] (1,-2)--(2,-2)--(2,-1)--(1,-1)--cycle;
\draw[fill, red, opacity=0.3] (1,0)--(2,0)--(2,1)--(1,1)--cycle;
\draw[fill, red, opacity=0.3] (1,2)--(2,2)--(2,4)--(1,4)--cycle;
\draw[fill, red, opacity=0.3] (-1,-1)--(0,-1)--(0,0)--(-1,0)--cycle;
\draw[fill, red, opacity=0.3] (-1,1)--(0,1)--(0,2)--(-1,2)--cycle;

\draw (0,-2)--(0,-1)--(-1,-1)--(-1,0)--(0,0)--(0,1)--(-1,1)--(-1,2)--(0,2)--(0,3)--(1,3)--(1,4)--(2,4)--(2,3)--(3,3);
\draw (1,-2)--(2,-2)--(2,-1)--(1,-1)--(1,0)--(2,0)--(2,1)--(1,1)--(1,2)--(2,2)--(2,1)--(3,1)--(3,2);

\draw (6,-2)--(6,-1)--(5,-1)--(5,0)--(6,0)--(6,1)--(5,1)--(5,2)--(6,2)--(6,3)--(7,3)--(7,4)--(8,4)--(8,3)--(9,3);
\draw (7,-2)--(8,-2)--(8,-1)--(7,-1)--(7,0)--(8,0)--(8,1)--(7,1)--(7,2)--(8,2)--(8,1)--(9,1)--(9,2);

\draw[fill, blue, opacity=0.3] (6,2)--(9,2)--(9,3)--(6,3)--cycle;
\draw[fill, blue, opacity=0.3] (6,0)--(8,0)--(8,1)--(6,1)--cycle;
\draw[fill, blue, opacity=0.3] (6,-2)--(8,-2)--(8,-1)--(6,-1)--cycle;
\draw[fill, red, opacity=0.3] (5,-1)--(7,-1)--(7,0)--(5,0)--cycle;
\draw[fill, red, opacity=0.3] (5,1)--(7,1)--(7,2)--(5,2)--cycle;
\draw[fill, red, opacity=0.3] (8,1)--(9,1)--(9,2)--(8,2)--cycle;
\draw[fill, red, opacity=0.3] (7,3)--(8,3)--(8,4)--(7,4)--cycle;

    \end{tikzpicture}
    \caption{A section of a wire and its optimal tilings: vertical (left) and horizontal (right).}
    \label{fig:wirePartition}
\end{figure}

As in Lingas' proof, we interpret the vertical tiling as setting the variable to \emph{true}, and the horizontal as \emph{false}. Lingas proved the following:

\begin{lemma}[Lingas~\cite{lingas1982power}]
\label{lemma:wire}
A minimum tiling with convex sets of a wire uses either vertical or horizontal rectangles but not both. Any other tiling requires at least one more convex set.
\end{lemma}

 The second tool is called a junction, and it serves to model a clause. Figure~\ref{fig:junction} depicts a junction corresponding to a clause of three literals. Figure~\ref{fig:junctionZoom} shows a zoom on the most important part of a junction. A junction has three arms, represented as dashed black line segments. A junction for a clause of two literals is obtained by blocking one of the arms of the junction. The blue line segments have length $1+\varepsilon$, for a fixed $\varepsilon$ arbitrarily small. Therefore, the red line segments are not aligned with the long black line segment to their right. A junction can be in four different orientations, which can be obtained successively by making rotations of $90^\circ$. Let us consider the orientation of the junction in Figure~\ref{fig:junction}. One wire is connected from the left, one from above, and one from the right. A wire can only be connected to a junction at one of its bends (see Figure~\ref{fig:wire}). We then remove the line segment corresponding to the arm of the junction, as illustrated in Figure~\ref{fig:junction}.
 
 \begin{figure}
    \centering
    \begin{tikzpicture}[rotate=270,scale=0.62]

\draw[fill, gray, opacity=0.3] (0,17)--(2,17)--(2,8)--(1,8)--(1,7)--(2,7)--(2,5)--(0.1,5)--(0.1,7)--(0.1,6)--(-1,6)--(-1,8)--(0,8)--cycle;

\draw[dashed, thick] (0,17)--(1,17);
\draw[dashed, thick] (0.1,5)--(1,5);
\draw[dashed, thick] (-1,7)--(-1,8);

\draw[dashed, thick, brown] (1,8)--(0,8);
\draw[dashed, thick, brown] (0.1,8)-- (0.1,7)--(1,7);

\draw [] (1,0)-- (1,1)-- (2,1)-- (2,2)-- (1,2)-- (1,3)-- (2,3)-- (2,4)-- (1,4)-- (1,5)-- (2,5)-- (2,7)-- (1,7)-- (1,8)-- (2,8)-- (2,17)-- (1,17)-- (1,18)-- (2,18)-- (2,19)-- (1,19)-- (1,20)-- (2,20)-- (2,21)-- (1,21)-- (1,22);
\draw [] (-1,22)-- (-1,21)-- (0,21)-- (0,20)-- (-1,20)--(-1,19)-- (0,19)-- (0,18)-- (-1,18)-- (-1,19)-- (-2,19)-- (-2,18)-- (-3,18)-- (-3,19)-- (-4,19)-- (-4,18);
\draw [thick,blue] (0.1,0)-- (-1,0);
\draw[thick,blue] (-1,1)-- (0.1,1);
\draw[thick,blue](0.1,2)-- (-1,2);
\draw[thick,blue] (-1,3)-- (0.1,3);
\draw[thick,blue] (0.1,4)-- (-1,4);
\draw[thick,red] (0.1,1)--(0.1,2);
\draw[thick,red] (0.1,3)--(0.1,4);
\draw (-1,2)--(-1,4);
\draw (-1,0)--(-1,1);
\draw (-1,3)-- (-2,3)-- (-2,4)-- (-3,4)-- (-3,3)-- (-4,3)-- (-4,4)-- (-5,4);
\draw [] (-5,5)-- (-5,6);
\draw [] (-4,6)-- (-4,5)--(-3,5)--(-3,6);
\draw [] (-2,6)-- (-2,5)--(-1,5)--(-1,7);
\draw [thick,red] (0.1,5)-- (0.1,7);
\draw [thick,blue](0.1,6)--(-1,6);
\draw (-1,7)-- (-2,7)-- (-2,6)-- (-3,6)-- (-3,7)-- (-4,7)-- (-4,6)-- (-5,6)-- (-5,7);
\draw [] (-4,16)-- (-4,17)-- (-3,17)-- (-3,16)-- (-2,16)-- (-2,17)-- (-1,17)-- (-1,16)-- (0,16);
\draw (0,17)-- (0,8)-- (-1,8);
\draw (0,9)--(-1,9)--(-1,10)--(0,10);
\draw (0,11)--(-1,11)--(-1,12)--(0,12);
\draw [] (-5,9)-- (-5,8)-- (-4,8)-- (-4,9)--(-3,9)-- (-3,8)-- (-2,8)-- (-2,9)-- (-3,9)-- (-3,10)-- (-2,10)-- (-2,11)-- (-3,11)-- (-3,12)--(-2,12);

    \end{tikzpicture}
    \caption{A junction for the MCT problem.}
    \label{fig:junction}
\end{figure}
 \begin{figure}
    \centering
    \begin{tikzpicture}[rotate=270,scale=0.9]
\useasboundingbox (-3,1.5) rectangle (2,10.5);
\draw[fill, gray, opacity=0.3] (0,17)--(2,17)--(2,8)--(1,8)--(1,7)--(2,7)--(2,5)--(0.2,5)--(0.2,7)--(0.2,6)--(-1,6)--(-1,8)--(0,8)--cycle;

\draw[dashed, thick] (0,17)--(1,17);
\draw[dashed, thick] (0.2,5)--(1,5);
\draw[dashed, thick] (-1,7)--(-1,8);

\draw[dashed, thick, brown] (1,8)--(0,8);
\draw[dashed, thick, brown] (0.2,8)-- (0.2,7)--(1,7);

\draw [] (1,0)-- (1,1)-- (2,1)-- (2,2)-- (1,2)-- (1,3)-- (2,3)-- (2,4)-- (1,4)-- (1,5)-- (2,5)-- (2,7)-- (1,7)-- (1,8)-- (2,8)-- (2,17)-- (1,17)-- (1,18)-- (2,18)-- (2,19)-- (1,19)-- (1,20)-- (2,20)-- (2,21)-- (1,21)-- (1,22);
\draw [] (-1,22)-- (-1,21)-- (0,21)-- (0,20)-- (-1,20)--(-1,19)-- (0,19)-- (0,18)-- (-1,18)-- (-1,19)-- (-2,19)-- (-2,18);
\draw [thick,blue] (0.2,0)-- (-1,0);
\draw[thick,blue] (-1,1)-- (0.2,1);
\draw[thick,blue](0.2,2)-- (-1,2);
\draw[thick,blue] (-1,3)-- (0.2,3);
\draw[thick,blue] (0.2,4)-- (-1,4);
\draw[thick,red] (0.2,1)--(0.2,2);
\draw[thick,red] (0.2,3)--(0.2,4);
\draw (-1,2)--(-1,4);
\draw (-1,0)--(-1,1);
\draw (-1,3)-- (-2,3)-- (-2,4)-- (-3,4);
\draw [] (-2,6)-- (-2,5)--(-1,5)--(-1,7);
\draw [thick,red] (0.2,5)-- (0.2,7);
\draw [thick,blue](0.2,6)--(-1,6);
\draw (-1,7)-- (-2,7)-- (-2,6)-- (-3,6);
\draw[very thick] (0,17)-- (0,8);
\draw (0,8)-- (-1,8);
\draw (0,9)--(-1,9)--(-1,10)--(0,10);
\draw (0,11)--(-1,11)--(-1,12)--(0,12);
\draw [] (-3,8)-- (-2,8)-- (-2,9)-- (-3,9)-- (-3,10)-- (-2,10)-- (-2,11)-- (-3,11)-- (-3,12)--(-2,12);

\fill[white] (-3.5,10.5)--(2.1,10.5)--(2.1,25)--(-3.5,25)--cycle;
\fill[white] (-3.5,1.5)--(2.1,1.5)--(2.1,-3)--(-3.5,-3)--cycle;

    \end{tikzpicture}
    \caption{A zoom on a junction for the MCT problem. The blue segments have length $1+\varepsilon$. One of the intersections between two dashed brown segments lies in the interior of the polygon. The red segments are not aligned with the long thick one to the right.}
    \label{fig:junctionZoom}
\end{figure}

If the tiling of the wire connected from the left is horizontal, then one of the rectangles can be prolonged into the junction. The same holds for the wire connected from the right. On the contrary, a rectangle can be prolonged from the wire connected from above only if the tiling is vertical. If a rectangle can be prolonged, we say that the wire \emph{sends true}, otherwise it \emph{sends false}. If a clause contains two negative literals $\bar{x},\bar{y}$ and one positive $z$, the corresponding junction is as in Figure~\ref{fig:junction}, or the $180^\circ$ rotation of it. The wire corresponding to $z$ is connected from above or below, and the wires corresponding to $x$ and $y$ are connected from the left and from the right, or vice versa. Therefore, the wire corresponding to $z$ sends \emph{true} if and only if $z$ is set to \emph{true}. On the contrary, the wire corresponding to $x$ (respectively $y$) sends \emph{true} if and only if $x$ (respectively $y$) is set to \emph{false}. If the clause has two positive literals, then the junction is vertical, and the junction behaves likewise.

Lingas proved that when minimising the number of convex polygons in a tiling, for each junction at least one adjacent wire sends \emph{true}. Before stating Lingas' lemma exactly, we need to explain the first step of the construction of the point set.



\subsection{Construction of the polygon with holes}

Let us consider one instance $(F,G)$ of MPLSAT. Lingas states that the planarity of $G$ implies that the junctions and the wires can be embedded as explained above, and so that they do not overlap~\cite{lingas1982power}. Thus we obtain a polygon with holes, that we denote by $\Pi$. He adds without proof that the dimensions of $\Pi$ are polynomially related to $|V|$, where $V$ denotes the vertex set of $G$. We show in the following paragraph how to embed the polygon with holes into a grid $\Lambda$, such that each edge consists of line segments of $\Lambda$. Actually, some parts of edges are not exactly line segments of $\Lambda$, but are shifted orthogonally by distance $\varepsilon$ (recall the red segments in Figure~\ref{fig:junction}). However, as $\varepsilon$ is arbitrarily small, this does not impact our proof. Thus, for sake of simplicity, we will from now on do as if the edges entirely consisted of segments of $\Lambda$. Moreover, we show how to construct $\Lambda$ in $\Theta(|V|^2)$ size. 

Let us consider an instance $(F,G)$ of MPLSAT. Recall that the vertex set of $G$ is the union of $X$ and $C$, where $X$ denotes the set of variables and $C$ the set of clauses. We define a new graph $G'=(V',E')$ as follow: For each vertex $x\in X$ of degree $\delta$, we have $\delta$ vertices $x_1, \dots, x_\delta$ in $V'$. Moreover for each $c\in C$ we have one vertex $c$ in $V'$. Now for each edge $(x,c)\in E$, we have one edge between $c\in V'$ and one of the $x_i\in V'$. We do so such that each vertex $x_i$ in $V'$ corresponding to some $x\in V$ has degree $1$. Then we add edges between the vertices $x_1,\dots, x_\delta$ so that they induce a path. We can do this such that the graph we obtain, $G'$, is still planar. Moreover, since a clause contains at most three vertices, the maximal degree of $G'$ is at most $3$. The number of vertices we have added is at most the number of edges in $G$, therefore $|V'|=\Theta(|V|)$. Following the result of Valiant~\cite{valiant1981universality}, we can embed $G'$ in a grid of size $\Theta(|V|^2)$, such that the edges consist of line segments of the grid. Let $s$ be the line segment incident to $c$ in the edge $\{x_i,c\}$, where $x_i$ corresponds a variable $x\in X$ and $c$ is a clause. Our constraint is that $x$ appears positively in $c$ if and only if $s$ is vertical. Moreover we impose $s$ to be of length at least $10$, so that we have enough space later to replace $c$ by a junction. We claim that we can find an embedding of $G$ on another grid, still of size $\Theta(|V|^2)$, that satisfies our constraint. We first scale the embedding by $3$. Then we can change the path of each edge adjacent to $c$ as illustrated in Figure~\ref{fig:clauseIncidence}, so that the embedding satisfies our first constraint. The line segments in red (respectively blue) correspond to edges $\{x,c\}$ where $x$ appears positively (respectively negatively) in $c$. By assumption, at most two variables appear positively, and at most two appear negatively. Therefore it is possible to adapt the paths of the edges so that the red line segments are vertical and the blue ones horizontal. Finally, we scale the grid by $10$ to satisfy the second constraint.

\begin{figure}
    \centering
    \begin{tikzpicture}[scale=1]
\draw[thick,red] (-1,0)--(0,0);
\draw[thick,blue] (0,1)--(0,0);
\draw[thick,blue] (0,-1)--(0,0);
\node at (0,0) {\textbullet};

\draw[right] (0,0) node[] {$c$};

\node at (6,0) {\textbullet};
\draw[thick,red] (6,1)--(6,0);
\draw[thick,blue] (7,0)--(6,0);
\draw[thick,blue] (5,0)--(6,0);
\draw (6,1)--(4,1)--(4,0)--(3,0);
\draw (7,0)--(7,2)--(6,2)--(6,3);
\draw (5,0)--(5,-2)--(6,-2)--(6,-3);
\draw[below] (6,0) node[] {$c$};

    \end{tikzpicture}
    \caption{How to adapt the embedding to satisfy the constraint. On the left, before adaptation, on the right, after adaptation.}
    \label{fig:clauseIncidence}
\end{figure}

Now we replace each clause by a junction. Let $x_1,\dots, x_\delta$ be the vertices corresponding to a variable $x\in X$. Observe that the set of line segments $\bigcup_{1 \leq i \leq \delta} \{e \in E' \mid x_i\in e\}$ is a tree. Let us take a constant sufficiently big, and scale the grid so that we can replace each tree corresponding to some variable $x\in X$ by a simple polygon that is as close as possible to the tree. Then, after a new scaling, we can replace each polygon by a wire, in a grid which is still of size $\Theta(|V|^2)$. As $\varepsilon$ introduced in the definition of the junctions is arbitrarily small, we can consider that all line segments of junctions and wires are segments of a grid. For sake of simplicity, we will from now on omit to mention that some line segments are not exactly on the grid. Again, it is possible to find an embedding in a grid of size $\Theta(|V|^2)$, such that wires do not intersect, and a wire is connected to a junction if and only the variable corresponding to the wire is contained in the clause corresponding to the junction.

We can now state Lingas' lemma:

\begin{lemma}[Lingas~\cite{lingas1982power}]
\label{lemma:junction}
In a minimum tiling with convex sets of $\Pi$, a junction contains wholly at least three convex sets. The junction contains wholly exactly three if and only if at least one of the wires connected to the junction sends \emph{true}.
\end{lemma}

\subsection{Discretisation of the line segments}

To construct the point set of the reduction from MPLSAT to MCT, we first construct a collection of line segments. We then discretise this collection by replacing each line segment by a set of collinear points.

Let us consider our polygon with holes $\Pi$ that lies in the grid $\Lambda$. The grid consists of points with integer coordinates, and line segments between points that are at distance $1$. We consider the collection of line segments consisting of $\Pi$ union each line segment of $\Lambda$ whose relative interior is not contained in the interior of $\Pi$. Notice that therefore we have line segments outside $\Pi$, but also inside its holes. Moreover, the collection of line segments that we obtain, denoted by $\Phi$, is a subgraph of the grid graph $\Lambda$.

Now we define $K$ as twice the number of unit squares in $\Lambda$ plus $1$. Finally, we replace each line segment in $\Phi$ by $K$ points evenly spaced. We denote this point set by $P$.

\subsection{Proof of correctness}

We have constructed $P$ in order to have the following property:
\begin{lemma}
\label{lemma:minimum}
In a minimum convex tiling $\Sigma$ of $P$, for each convex set $S \in \Sigma$, the interior of $S$ does not intersect $\Phi$.
\end{lemma}

Before proving Lemma~\ref{lemma:minimum}, we first explain how we use it. Let $K'$ denote the number of unit squares in $\Phi$, plus the minimum number of rectangles in a partition of the wires, plus three times the number of clauses. Using Lemmas~\ref{lemma:wire} and~\ref{lemma:junction} shown by Lingas coupled with Lemma~\ref{lemma:minimum}, we immediately obtain the following theorem:

\begin{theorem}\label{equivalence}
The formula $F$ is satisfiable if and only if there exists a convex tiling of $P$ with $K'$ polygons.
\end{theorem}

Since $P$ and $K'$ can be computed in polynomial time, Theorem~\ref{equivalence} implies Theorem~\ref{NPC} for the MCT problem. We use a packing argument, and claim that in a convex tiling $\Sigma$ of $P$, if a convex set $S\in \Sigma$ has large area, then most of its area is contained in a unique cell of $\Phi$. Then we show that in a minimum convex tiling, all convex sets have large area, and that each of them fills the cell that contains it. For a set $S$, let $\A(S)$ be the area of $S$.



\begin{lemma}
\label{lemma:bigAreaNotCross}
Let $L$ and $L'$ be two squares in $\Lambda$, and $S$ be a convex polygon whose interior does not contain any point in $P$. If $\A(S \cap L)>1/K$, and the boundary of $S$ crosses a line segment of $\Phi$ between $L$ and $L'$, then $\A(S \cap L')\leq 1/K$.
\end{lemma}

\begin{proof}
The proof is illustrated in Figure~\ref{fig:bigAreaNotCross}. By assumption, $S$ intersects a line segment whose endpoints $p$ and $q$ are at distance $1/K$. Let us consider the two line segments $s$ and $s'$ of the boundary of $S$ that intersect the line $\ell$ which contains $p$ and $q$. Assume for contradiction that the lines containing respectively $s$ and $s'$ do not intersect, or intersect on the side of $\ell$ where $L$ lies. This implies that $S \cap L$ is contained in a parallelogram that has area $1/K$, as illustrated in Figure~\ref{fig:parallelogramTrapezoid}. Indeed such a parallelogram has base $1/K$ and height $1$, therefore $\A(S \cap L)\leq 1/K$. This shows that the lines containing respectively $s$ and $s'$ intersect on the side of $\ell$ where $L'$ lies. Using the same arguments as above, this implies $\A(S\cap L')\leq 1/K$.
\end{proof}

\begin{figure}[ht]
    \centering
    \begin{tikzpicture}[scale=8]

\draw[fill, blue, opacity=0.3] (0.64,0.25)--(-0.2,0.19)--(-0.18,0.23)--(0.6,0.52)--cycle;

\draw (0.64,0.25)--(-0.2,0.19);
\draw (0.6,0.52)--(-0.18,0.23);

\node at (0,-0.1) {\textbullet};
\node at (0,0) {\textbullet};
\node at (0,0.1) {\textbullet};
\node at (0,0.2) {\textbullet};
\node at (0,0.3) {\textbullet};
\node at (0,0.4) {\textbullet};
\node at (0,0.5) {\textbullet};

\node at (-0.2,0) {\textbullet};
\node at (-0.1,0) {\textbullet};
\node at (0.1,0) {\textbullet};
\node at (0.2,0) {\textbullet};
\node at (0.3,0) {\textbullet};
\node at (0.4,0) {\textbullet};
\node at (0.5,0) {\textbullet};
\node at (0.6,0) {\textbullet};
\node at (0.7,0) {\textbullet};

\draw (0.3,0.1) node[] {$L$};
\draw[blue] (0.5,0.36) node[] {$S$};
\draw (-0.15,0.28) node[] {$s$};
\draw (-0.15,0.15) node[] {$s'$};

\draw[above right] (0,0.3) node[] {$p$};
\draw[below right] (0,0.2) node[] {$q$};

    \end{tikzpicture}
    \caption{If $A(S \cap L)>1/K$, the two lines containing $s$ and $s'$ intersect on the left side.}
    \label{fig:bigAreaNotCross}
\end{figure}
\begin{figure}
    \centering
    \begin{tikzpicture}[scale=4]

\draw (0,0)--(0,1)--(1,1)--(1,0)--cycle;

\node at (0,0.2) {\textbullet};
\node at (0,0.4) {\textbullet};

\draw[fill, blue, opacity=0.3] (0,0.2)--(0,0.4)--(1,0.8)--(1,0.6)--cycle;
\draw[fill, red, opacity=0.3] (0,0.2)--(0,0.4)--(1,1.3)--(1,1.1)--cycle;

\draw[left] (0,0.4) node[] {$p$};
\draw[left] (0,0.2) node[] {$q$};
\draw[] (0,1) node[below right] {$L$};

    \end{tikzpicture}
    \caption{The area of the parallelograms is $1/K$.}
    \label{fig:parallelogramTrapezoid}
\end{figure}

\begin{lemma}
\label{lemma:notSameLineColumn}
Let $R$ be a rectilinear polygon on $\Lambda$ whose interior does not contain any point in $P$. Let $\Psi$ be a set of squares of $\Lambda$ contained in $R$, such that no two of them are on the same row or same column of $\Lambda$. Let $\Sigma$ be a minimum convex tiling of $P$. Let $S$ be an element of $\Sigma$, such that for each $L_i \in \Psi$, we have $\A(S \cap L_i)>1/K$. Then there exist $|\Psi|-1$ squares $\{L_i\}$ in $\Psi$ and $|\Psi|-1$ convex sets $\{S_i\}$ in $\Sigma$, such that for any $i$: $\A(S_i \cap L_i)>1/K$.
\end{lemma}

\begin{proof}
We assume $|\Psi| \geq 2$, otherwise there is nothing to prove. Using Lemma~\ref{lemma:bigAreaNotCross}, we know that if such a $S$ exists, then there are no line segment of $\Phi$ between any two squares in $\Psi$. We can observe thanks to how wires and junctions are defined that $|\Psi|$ is at most three. Moreover for the same reason, there are $|\Psi|-1$ squares in $\Psi$ such that the area of their intersection with $S$ is at most $1/2$, as illustrated in Figure~\ref{fig:notSameLineColumn}. Observe that by taking each unit square in $\Lambda$ as a convex face, we obtain a convex partition with $\frac{K-1}{2}$ convex faces. In particular, there are at most $K$ convex faces in $\Sigma$, which implies that for each unit square $L \in \Psi$, there exists a convex face $\tilde{S} \in \Sigma$ with $\A(\tilde{S} \cap L)>1/K$. Moreover by iterating this process we can choose these convex sets so that they are distinct.
\end{proof}

\begin{figure}[ht]
    \centering
    \begin{tikzpicture}[scale=1.05]

\draw[blue] (-0.55,1)--(1.73,4);
\draw[fill, blue, opacity=0.3] (1.73,4)--(2,3.7)--(1.78,2.88)--(1,2)--(-0.55,1)--cycle;

\draw[dashed, red] (0,1)--(0,2)--(1,2)--(1,3)--(2,3);

\draw (0,0)--(0,1)--(-1,1)--(-1,2)--(0,2)--(0,3)--(1,3)--(1,4)--(2,4)--(2,3)--(3,3);
\draw (1,0)--(2,0)--(2,1)--(1,1)--(1,2)--(2,2)--(2,1)--(3,1)--(3,2);

\draw[red] (1.2,3.75) node[] {$L_1$};
\draw[red] (0.3,2.6) node[] {$L_2$};
\draw[red] (-0.6,1.6) node[] {$L_3$};
\draw[blue] (1.73,2.31) node[] {$S$};

    \end{tikzpicture}
    \caption{If $A(S \cap L_1)\geq1/2$, then $A(S \cap L_2) \leq 1/2$ and $A(S \cap L_3)\leq 1/2$.}
    \label{fig:notSameLineColumn}
\end{figure}

\begin{proof}[Proof of Lemma~\ref{lemma:minimum}]
We define a map $f_0$ that associates each square $L$ in $\Lambda$ to a convex set $S \in \Sigma$ such that $\A(S\cap L)>1/K$. Such a set exists since otherwise we would need more than $K$ sets only to fill the square $L$. Now we define a map $f$ in several steps. We start from $f_0$. At step $i+1$, find $L$ and $L'$ on the same line or the same column, and $L''$ between them, such that $f_i(L)=f_i(L')\neq f_i(L'')$. If this is not possible then stop. Otherwise define $f_{i+1}$ as identically equal to $f_i$, except for all square between $L$ and $L'$, that are mapped to $f_i(L)$. Notice that $f_{i+1}$ keeps the property that for each square $\tilde{L}$, we have $\A(f_i(\tilde{L})\cap \tilde{L})>1/K$. Moreover, this procedure must stop eventually because the interiors of the convex sets are non-overlapping. We denote by $f$ the map obtained after the last iteration.

We denote by $f^{-1}$ the function that maps a set $S \in \Sigma$ to $\{L \in \Lambda \mid f(L)=S\}$. By Lemmas~\ref{lemma:bigAreaNotCross} and~\ref{lemma:notSameLineColumn}, if for some $S$ there are $m$ squares in $f^{-1}(S)\geq 2$ such that no pair of squares are on a line nor on a column, then there exists at least $m-1$ convex sets $S_i$ such that for each $S_i$, $f^{-1}(S_i)=\emptyset$. Moreover, such a $S_i$ cannot appear when considering another element of $\Sigma$: If Lemma~\ref{lemma:notSameLineColumn} when applied to $S$ gives the existence of the $\{S_i\}$, and when applied to $S'\neq S$ gives the existence of the $\{S'_j\}$, then all these convex sets are distinct. We denote by $M$ the number of such sets, taken over all sets $S\in \Sigma$ with $f^{-1}(S)\geq 2$.

We are going to partition $P$ into at most $|\Sigma| -M$ rectilinear polygons (not necessarily convex). We denote by $T$ the partition. For each set $S \in \Sigma$ such that $|f^{-1}(S)|= 1$, we add in $T$ the square in $f^{-1}(S)$. For each $S \in \Sigma$ such that $|f^{-1}(S)|\geq 1$, we add in $T$ the rectilinear polygon consisting of the union of the squares in $f^{-1}(S)$. We know thanks to Lemma~\ref{lemma:bigAreaNotCross} that the boundary of this rectilinear polygon does not cross any line segment in $\Phi$. Following what was argued above, there are at most $|\Sigma| -M$ sets in~$T$. We now construct a new convex tiling $\Sigma'$, that we claim to be minimum. We replace each set in $T$ that is not convex by rectangles. From the construction of wires and junctions, there exists a partition of the sets in $T$ into rectangles that will add $M$ new sets. Looking at the construction of $T$ and then $\Sigma'$, we can observe that the convex tiling $\Sigma$ was minimum if and only if the interior of all sets $S$ does not intersect $\Phi$.
\end{proof}

\subsection{How to adapt the proof to MCP}

Let us explain how to adapt the proof to the case of the MCP problem. Notice that the only lemma for which it makes a difference is Lemma~\ref{lemma:junction}. Indeed, if in Figure~\ref{fig:junction} the wire connected from below sends \emph{true}, then one can prolong a rectangle of this wire so that it contains the small rectangle in the junction. Therefore only three rectangles are wholly contained in the junction. However, the vertex on the top left of the prolonged rectangle is not a point of $P$. This can be solved by adding a line segment to the construction, coloured in red in Figure~\ref{fig:junctionMCP}. With this construction, the rectangle is prolonged into a trapezoid (delimited above by the dashed red line segment). The convex set above remains a rectangle, and the one to the left becomes a convex quadrilateral with a right angle. Now one can adapt the proof by taking into account this new triangular hole that has been inserted in each junction. The proof can then be done similarly to the one of the MCT problem.

\begin{figure}
    \centering
    \begin{tikzpicture}[scale=0.9]

\draw[fill, gray, opacity=0.3] (0,17)--(2,17)--(2,8)--(1,8)--(1,7)--(2,7)--(2,5)--(0.1,5)--(0.1,7)--(-1,6)--(-1,8)--(0,8)--cycle;

\draw[dashed, thick] (0,17)--(1,17);
\draw[dashed, thick] (0.1,5)--(1,5);
\draw[dashed, thick] (-1,7)--(-1,8);

\draw[dashed, thick, brown] (1,8)--(0,8);
\draw[dashed, thick, brown] (0.1,8)-- (0.1,7)--(1,7);

\draw[thick,red] (-1,6)--(0.1,7);
\draw[thick,red,dashed] (0.1,7)--(1,8);

\draw [] (1,0)-- (1,1)-- (2,1)-- (2,2)-- (1,2)-- (1,3)-- (2,3)-- (2,4)-- (1,4)-- (1,5)-- (2,5)-- (2,7)-- (1,7)-- (1,8)-- (2,8)-- (2,17)-- (1,17)-- (1,18)-- (2,18)-- (2,19)-- (1,19)-- (1,20)-- (2,20)-- (2,21)-- (1,21)-- (1,22);
\draw [] (-1,22)-- (-1,21)-- (0,21)-- (0,20)-- (-1,20)--(-1,19)-- (0,19)-- (0,18)-- (-1,18)-- (-1,19)-- (-2,19)-- (-2,18)-- (-3,18)-- (-3,19)-- (-4,19)-- (-4,18);
\draw [] (0.1,0)-- (-1,0);
\draw[] (-1,1)-- (0.1,1);
\draw[](0.1,2)-- (-1,2);
\draw[] (-1,3)-- (0.1,3);
\draw[] (0.1,4)-- (-1,4);
\draw[] (0.1,1)--(0.1,2);
\draw[] (0.1,3)--(0.1,4);
\draw (-1,2)--(-1,4);
\draw (-1,0)--(-1,1);
\draw (-1,3)-- (-2,3)-- (-2,4)-- (-3,4)-- (-3,3)-- (-4,3)-- (-4,4)-- (-5,4);
\draw [] (-5,5)-- (-5,6);
\draw [] (-4,6)-- (-4,5)--(-3,5)--(-3,6);
\draw [] (-2,6)-- (-2,5)--(-1,5)--(-1,7);
\draw [] (0.1,5)-- (0.1,7);
\draw [](0.1,6)--(-1,6);
\draw (-1,7)-- (-2,7)-- (-2,6)-- (-3,6)-- (-3,7)-- (-4,7)-- (-4,6)-- (-5,6)-- (-5,7);
\draw [] (-4,16)-- (-4,17)-- (-3,17)-- (-3,16)-- (-2,16)-- (-2,17)-- (-1,17)-- (-1,16)-- (0,16);
\draw (0,17)-- (0,8)-- (-1,8);
\draw (0,9)--(-1,9)--(-1,10)--(0,10);
\draw (0,11)--(-1,11)--(-1,12)--(0,12);
\draw [] (-5,9)-- (-5,8)-- (-4,8)-- (-4,9)--(-3,9)-- (-3,8)-- (-2,8)-- (-2,9)-- (-3,9)-- (-3,10)-- (-2,10)-- (-2,11)-- (-3,11)-- (-3,12)--(-2,12);

    \end{tikzpicture}
    \caption{A junction for the MCP problem.}
    \label{fig:junctionMCP}
\end{figure}

\section{NP-hardness of CPNCS}
\label{sec:NPhardCPNCS}

Let us consider the CPNCS problem with additional constraints on the input point set $P$. We consider the CPNCS problem on a $3$-directed point set $P$ with set of directions $L$, such that there are no five collinear points, and a point $p\in P$ is contained in at most one line in $L$. We show that CPNCS is NP-hard, even for such constrained point sets. 

\begin{figure}[ht]
    \centering
    \begin{tikzpicture}[scale=0.4]

\draw  (-16,8)-- (-9,8);
\draw[thick,red]  (-11,9)-- (-11,3);
\draw  (-12,4)-- (0,4);
\draw[thick,red]  (-11,0)-- (-2,9);
\draw  (-4,9)-- (-4,-1);
\draw  (2.33,8)-- (2.67,8);
\draw  (8.61,8)-- (8.27,8);
\draw[thick,red]  (7,8.6)-- (7,3.32);
\draw[red,thick]  (15.805,8.805)-- (7.295,0.295);
\draw  (6.37,4)-- (6.71,4);
\draw  (14,8.6)-- (14,8.24);
\draw  (17.27,4)-- (17.65,4);
\draw  (14,-0.58)-- (14,-0.24);

\node at (2.33,8){\textbullet};
\node at (2.67,8){\textbullet};
\node at (7,8.6){\textbullet};
\node at (7,8.3){\textbullet};
\node at (8.61,8){\textbullet};
\node at (8.27,8){\textbullet};
\node at (7,3.32){\textbullet};
\node at (7,3.64){\textbullet};
\node at (6.37,4){\textbullet};
\node at (6.71,4){\textbullet};
\node at (7.295,0.295){\textbullet};
\node at (7.665,0.665){\textbullet};
\node at (14,8.6){\textbullet};
\node at (14,8.24){\textbullet};
\node at (14,-0.58){\textbullet};
\node at (14,-0.24){\textbullet};
\node at (17.27,4){\textbullet};
\node at (17.65,4){\textbullet};
\node at (15.805,8.805){\textbullet};
\node at (15.455,8.455){\textbullet};

    \end{tikzpicture}
    \caption{The reduction from Maximum independent set in segments to CPNCS. A segment in an independent set corresponds to a segment in a covering which covers four points instead of only two.}
    \label{fig:NPCPNCS}
\end{figure}

Kratochv{\'\i}l and Ne{\v{s}}et{\v{r}}il have shown that \emph{Maximum Independent Set in Intersection Graphs of Segments} is NP-hard, even when the segments lie in only three possible directions~\cite{kratochvil1990independent}. The problems corresponds to finding a maximum set of pairwise non-intersecting segments. There is an additional constraint that no two parallel segments intersect. We do a reduction from this problem to CPNCS with the additional constraints. The reductions is illustrated in Figure~\ref{fig:NPCPNCS}. Let us consider a set of such line segments. We first shift the segments so that no two segments are aligned. Since no two parallel segments intersect, this does not change the intersection graph. We also extend some segments such that if two segments intersect, then the intersection is in the relative interiors of the segments. We replace each segment by four collinear points, that lie on the segment. Let us consider a segment $s$ with endpoints $u$ and $v$. By construction, there is a connected subset of $s$ that contains $u$ but no intersection point between $s$ and another segment. We add two points in this connected subset. The same holds with $v$, and we add two other points in the corresponding connected subset of $s$ that contains $v$. We do it such that three points are collinear only if they lie on the same segment of the input set. It is clear that the point set satisfies our constraints.

We denote by $n$ the number of segments. We claim that there is an independent set of $s$ segments if and only if there is a covering of the points with $2n-s$ non-crossing segments. Assume there is an independent set of $s$ segments, and let us show the existence of a covering with $2n-s$ elements. For each segment in the independent set, we take it in the covering. For the other segments, we take two segments to cover the four points: one for each pair of close points. Therefore we have covered the points with $s+2(n-s)$ segments, which are non-crossing by assumption.

Let us now assume that there is a covering of the point set with $2n-s$ non-crossing segments. Let us denote by $x$ the number of segments that cover more than two points. Such a segment covers at most four points, and there are $4n$ points to cover. Hence, we have $4x+2(2n-s-x)= 4n$, which implies $x= s$. By construction, each segment that covers more than two points is a segment of the input set, thus there exists an independent set of $s$ segments.

\section{Proof of Theorem~\ref{thm:lowerBoundGreedy}}
\label{sec:badPointSet}

\begin{figure}
    \centering
    \begin{tikzpicture}[scale=0.9]
\draw (0,0)-- (0,4.813387995947737);
\draw (1,0)-- (1,5.221758236089198);
\draw (2,0)-- (2,4.734855257458994);
\draw (3,0)-- (3,5.598715380835154);
\draw (4,0)-- (4,4.813387995947737);
\draw (5,0)-- (5,5.425943356159921);
\draw (6,0)-- (6,4.624909423574758);
\draw (8,2)-- (14,2);
\draw (8,0)-- (14,0);
\draw (8,4)-- (14,4);
\draw (8,1.3422409547453562)-- (9,0.6668594037421736);
\draw (10,1.295121311652111)-- (11,1.640665361002569);
\draw (12,0.5569135698579346)-- (13,1);
\draw (8,3.3526790600571448)-- (9,2.7244171521472067);
\draw (10,3.4312117985458865)-- (11,3.1799070353819108);
\draw (12,3.6668100140121127)-- (13,2.787243342938204);
\draw (14,3.3526790600571417)-- (13,2.787243342938204);
\draw (14,0.808218333021909)-- (13,1);
\draw (8,4.813387995947742)-- (9,5.221758236089198);
\draw (10,4.734855257458994)-- (11,5.598715380835154);
\draw (12,4.813387995947737)-- (13,5.425943356159921);
\draw (14,4.624909423574758)-- (13,5.425943356159921);

\node at   (0,0) {\textbullet};
\node at   (1,0) {\textbullet};
\node at   (2,0) {\textbullet};
\node at   (3,0) {\textbullet};
\node at   (4,0) {\textbullet};
\node at   (5,0) {\textbullet};
\node at   (6,0) {\textbullet};
\node at   (0,2) {\textbullet};
\node at   (9,2) {\textbullet};
\node at   (10,2) {\textbullet};
\node at   (11,2) {\textbullet};
\node at   (12,2) {\textbullet};
\node at   (13,2) {\textbullet};
\node at   (14,2) {\textbullet};
\node at   (0,4) {\textbullet};
\node at   (9,4) {\textbullet};
\node at   (10,4) {\textbullet};
\node at   (11,4) {\textbullet};
\node at   (12,4) {\textbullet};
\node at   (13,4) {\textbullet};
\node at   (14,4) {\textbullet};
\node at   (0,1.342240954745354) {\textbullet};
\node at   (9,0.6668594037421736) {\textbullet};
\node at   (10,1.295121311652111) {\textbullet};
\node at   (11,1.640665361002569) {\textbullet};
\node at   (12,0.5569135698579346) {\textbullet};
\node at   (13,1) {\textbullet};
\node at   (14,0.808218333021909) {\textbullet};
\node at   (0,3.36838560775489) {\textbullet};
\node at   (9,2.7244171521472067) {\textbullet};
\node at   (10,3.4312117985458865) {\textbullet};
\node at   (11,3.1799070353819108) {\textbullet};
\node at   (12,3.6668100140121127) {\textbullet};
\node at   (13,2.787243342938204) {\textbullet};
\node at   (14,3.3526790600571417) {\textbullet};
\node at   (0,4.813387995947737) {\textbullet};
\node at   (9,5.221758236089198) {\textbullet};
\node at   (10,4.734855257458994) {\textbullet};
\node at   (11,5.598715380835154) {\textbullet};
\node at   (12,4.813387995947737) {\textbullet};
\node at   (13,5.425943356159921) {\textbullet};
\node at   (14,4.624909423574758) {\textbullet};
\node at   (8,0) {\textbullet};
\node at   (9,0) {\textbullet};
\node at   (10,0) {\textbullet};
\node at   (11,0) {\textbullet};
\node at   (12,0) {\textbullet};
\node at   (13,0) {\textbullet};
\node at   (14,0) {\textbullet};
\node at   (8,2) {\textbullet};
\node at   (8,4) {\textbullet};
\node at   (8,1.3422409547453562) {\textbullet};
\node at   (8,4.813387995947742) {\textbullet};
\node at   (8,3.3526790600571448) {\textbullet};
\node at   (1,2) {\textbullet};
\node at   (2,2) {\textbullet};
\node at   (3,2) {\textbullet};
\node at   (4,2) {\textbullet};
\node at   (5,2) {\textbullet};
\node at   (6,2) {\textbullet};
\node at   (1,4) {\textbullet};
\node at   (2,4) {\textbullet};
\node at   (3,4) {\textbullet};
\node at   (4,4) {\textbullet};
\node at   (5,4) {\textbullet};
\node at   (6,4) {\textbullet};
\node at   (1,0.6668594037421736) {\textbullet};
\node at   (2,1.295121311652111) {\textbullet};
\node at   (3,1.6406653610025688) {\textbullet};
\node at   (4,0.5569135698579346) {\textbullet};
\node at   (5,1) {\textbullet};
\node at   (6,0.808218333021909) {\textbullet};
\node at   (1,2.7244171521472067) {\textbullet};
\node at   (2,3.4312117985458865) {\textbullet};
\node at   (3,3.1799070353819108) {\textbullet};
\node at   (4,3.6668100140121127) {\textbullet};
\node at   (5,2.771536795240456) {\textbullet};
\node at   (6,3.3526790600571417) {\textbullet};
\node at   (1,5.221758236089198) {\textbullet};
\node at   (2,4.734855257458994) {\textbullet};
\node at   (3,5.598715380835154) {\textbullet};
\node at   (4,4.813387995947737) {\textbullet};
\node at   (5,5.425943356159921) {\textbullet};
\node at   (6,4.624909423574758) {\textbullet};
    \end{tikzpicture}
    \caption{On the left a minimum covering of the point set. On the right a covering given by the greedy algorithm. As the horizontal segments contain seven points, they are first taken by the greedy algorithm. But afterwards many segments are needed to cover the remaining points.}
    \label{fig:greedy}
\end{figure}
The construction is illustrated in Figure~\ref{fig:greedy}. Consider the set of $2k(2k+1)$ points $P:=\{(x,y) \mid x,y \in \mathbb{N},1\leq x \leq 2k+1,1\leq y \leq 2k\}$. Now we perturb slightly the points such that any line that covers at least three points is either vertical or horizontal. Actually, after perturbation, the lines might not be exactly vertical or horizontal, but we keep referring to them as such for simplicity. Moreover, we perturb again the points such that only $k$ horizontal lines cover at least three points. We do so such that, for any pair of consecutive horizontal lines that covers at least three points, there are $k+1$ points in between (which correspond to the points of the former horizontal line). In the end, a vertical line covers at least three points if and only if it covers exactly $2k$ points. Similarly, a horizontal line covers at least three points if and only if it covers exactly $2k+1$ points. We denote this new point set by $P'$, and denote by $n$ the cardinality of $P'$. It is clear that one can cover all points using $2k+1$ vertical segments (which are thus non-crossing). However, the greedy algorithm will first pick the $k$ horizontal segments with even $y$-coordinate. Then it remains to cover the $k(2k+1)=\frac{n}{2}$ remaining points, but there is no line covering three of these points. Therefore the greedy algorithm picks at least $\frac{k(2k+1)}{2}$ more segments. Hence, the approximation ratio is at least $\frac{k(2k+1)/2}{2k+1}=k/2=\Omega(\sqrt{n})$.

Recall that in Theorem~\ref{thm:approxAlgo}, we introduce another algorithm whose approximation ratio is $\mathcal{O}(\sqrt{n}\log n)$ for solving CPNCS. The algorithm first computes a minimum covering with lines of the point set, and then divides these lines into non-crossing segments. By the same reasoning, this algorithm will also have an approximation ratio in $\Omega(\sqrt{n})$ on the point set $P'$.

\section{Open problems}
It would be interesting to have approximation algorithms for MCP, MCT and CPNCS with better ratio than $\mathcal{O}(\log OPT)$. As MCP is not APX-hard unless $\textit{NP}\subseteq \textit{DTIME}(2^{\textit{polylog}~n})$~\cite{bandyapadhyay2014approximation}, we expect that some improvement can be achieved.

A natural question is to ask whether MCP is FPT with respect to the number of faces in an optimal convex partition, as we have only shown a constant-approximation FPT algorithm. This question only makes sense when having several points on a line is allowed, since otherwise the minimum number of convex faces is linear in the number of inner points, and then the answer is known~\cite{grantson2004fixed}.

We have shown that the decision versions of MCP and CPNCS are NP-complete, and that the one of MCT is NP-hard, but the question whether the decision version of MCT is in NP remains open. We also do not know the complexity of MCP and MCT when it is assumed that no three points are collinear.



\newpage

\bibliography{references}

\end{document}